\documentclass{bmcart}

\usepackage[utf8]{inputenc} %unicode support
\usepackage{graphicx,subfigure,mathrsfs,amsmath,amsthm,amssymb,algorithm2e,url,color}

\newcommand{\tabincell}[2]{\begin{tabular}{@{}#1@{}}#2\end{tabular}}
\newtheorem{definition}{Definition}

\newtheorem{lemma}{Lemma}

\def\X{\mathscr{X}}

\startlocaldefs

\endlocaldefs

\begin{document}

%%% Start of article front matter
\begin{frontmatter}

\begin{fmbox}
\dochead{Research}


\title{In the Light of Deep Coalescence: Revisiting Trees Within Networks}

%%%%%%%%%%%%%%%%%%%%%%%%%%%%%%%%%%%%%%%%%%%%%%
%%                                          %%
%% Enter the authors here                   %%
%%                                          %%
%% Specify information, if available,       %%
%% in the form:                             %%
%%   <key>={<id1>,<id2>}                    %%
%%   <key>=                                 %%
%% Comment or delete the keys which are     %%
%% not used. Repeat \author command as much %%
%% as required.                             %%
%%                                          %%
%%%%%%%%%%%%%%%%%%%%%%%%%%%%%%%%%%%%%%%%%%%%%%

\author[
   addressref={aff1},
   email={jiafan.zhu@rice.edu}
]{\inits{JZ}\fnm{Jiafan} \snm{Zhu}}
\author[
   addressref={aff1},                   % id's of addresses, e.g. {aff1,aff2}
   %corref={aff1},                       % id of corresponding address, if any
%   noteref={n1},                        % id's of article notes, if any
   %email={\{yy9,nakhleh\}@rice.edu}   % email address
]{\inits{YY}\fnm{Yun} \snm{Yu}}
\author[
   addressref={aff1,aff2},
   corref={aff1},  
   email={nakhleh@rice.edu}
]{\inits{LN}\fnm{Luay} \snm{Nakhleh}}

%%%%%%%%%%%%%%%%%%%%%%%%%%%%%%%%%%%%%%%%%%%%%%
%%                                          %%
%% Enter the authors' addresses here        %%
%%                                          %%
%% Repeat \address commands as much as      %%
%% required.                                %%
%%                                          %%
%%%%%%%%%%%%%%%%%%%%%%%%%%%%%%%%%%%%%%%%%%%%%%

\address[id=aff1]{%                           % unique id
  \orgname{Department of Computer Science, Rice University}, % university, etc
 % \street{6100 Main Street},                     %
  \city{Houston},                              % city
  \postcode{77005}                                % post or zip code
  \cny{Texas, USA}                                    % country
}

\address[id=aff2]{%                           % unique id
  \orgname{Department of BioSciences, Rice University}, % university, etc
 % \street{6100 Main Street},                     %
  \city{Houston},                              % city
  \postcode{77005}                                % post or zip code
  \cny{Texas, USA}                                    % country
}

%%%%%%%%%%%%%%%%%%%%%%%%%%%%%%%%%%%%%%%%%%%%%%
%%                                          %%
%% Enter short notes here                   %%
%%                                          %%
%% Short notes will be after addresses      %%
%% on first page.                           %%
%%                                          %%
%%%%%%%%%%%%%%%%%%%%%%%%%%%%%%%%%%%%%%%%%%%%%%

%%\note{Sample of title note}     % note to the article
%\note[id=n1]{Equal contributor} % note, connected to author

\end{fmbox}% comment this for two column layout

%%%%%%%%%%%%%%%%%%%%%%%%%%%%%%%%%%%%%%%%%%%%%%
%%                                          %%
%% The Abstract begins here                 %%
%%                                          %%
%% Please refer to the Instructions for     %%
%% authors on http://www.biomedcentral.com  %%
%% and include the section headings         %%
%% accordingly for your article type.       %%
%%                                          %%
%%%%%%%%%%%%%%%%%%%%%%%%%%%%%%%%%%%%%%%%%%%%%%

\begin{abstractbox}

\begin{abstract} % abstract
\parttitle{Background} %if any
 Phylogenetic networks model reticulate evolutionary histories. The last two decades have seen an
 increased interest in establishing mathematical results and developing computational methods 
 for inferring and analyzing these networks.  A salient concept underlying a great majority of these 
 developments has been the notion that a network displays a set of trees and those trees can be 
 used to infer, analyze, and study the network. 
\parttitle{Results} %if any
In this paper, we show that in the presence of coalescence effects, the set of displayed trees is not 
sufficient to capture the network. We formally define the set of parental trees of a network and make 
three contributions based on this definition. First, we extend the notion of anomaly zone to phylogenetic 
networks and report on anomaly results for different networks. 
 Second, we demonstrate how coalescence 
events could negatively affect the ability to infer a species tree that could be augmented into the 
correct network. Third, we demonstrate how a phylogenetic network can be viewed as a mixture model 
that lends itself to a novel inference approach via gene tree clustering.  
\parttitle{Conclusions} %if any
Our results demonstrate the limitations of focusing on the set of trees displayed by a network when 
analyzing and inferring the network. Our findings can form the basis for achieving higher accuracy 
when inferring phylogenetic networks and open up new venues for research in this area, including new 
problem formulations based on the notion of a network's parental trees. 

\end{abstract}

%%%%%%%%%%%%%%%%%%%%%%%%%%%%%%%%%%%%%%%%%%%%%%
%%                                          %%
%% The keywords begin here                  %%
%%                                          %%
%% Put each keyword in separate \kwd{}.     %%
%%                                          %%
%%%%%%%%%%%%%%%%%%%%%%%%%%%%%%%%%%%%%%%%%%%%%%

%\begin{keyword}
%\kwd{sample}
%\kwd{article}
%\kwd{author}
%\end{keyword}

% MSC classifications codes, if any
%\begin{keyword}[class=AMS]
%\kwd[Primary ]{}
%\kwd{}
%\kwd[; secondary ]{}
%\end{keyword}

\end{abstractbox}
%
%\end{fmbox}% uncomment this for twcolumn layout

\end{frontmatter}

%%%%%%%%%%%%%%%%%%%%%%%%%%%%%%%%%%%%%%%%%%%%%%
%%                                          %%
%% The Main Body begins here                %%
%%                                          %%
%% Please refer to the instructions for     %%
%% authors on:                              %%
%% http://www.biomedcentral.com/info/authors%%
%% and include the section headings         %%
%% accordingly for your article type.       %%
%%                                          %%
%% See the Results and Discussion section   %%
%% for details on how to create sub-sections%%
%%                                          %%
%% use \cite{...} to cite references        %%
%%  \cite{koon} and                         %%
%%  \cite{oreg,khar,zvai,xjon,schn,pond}    %%
%%  \nocite{smith,marg,hunn,advi,koha,mouse}%%
%%                                          %%
%%%%%%%%%%%%%%%%%%%%%%%%%%%%%%%%%%%%%%%%%%%%%%

%%%%%%%%%%%%%%%%%%%%%%%%% start of article main body
% <put your article body there>

%%%%%%%%%%%%%%%%
%% Background %%
%%

\vspace{-.3in}
\section*{Background}
 Evolutionary, or explicit, phylogenetic networks are graphical models that model reticulate evolutionary
 histories \cite{HusonBryant2006,nakhleh2010evolutionary,bapteste2013networks}. Such evolutionary histories arise when processes such as horizontal gene 
 transfer or hybridization occur. Research into mathematical properties, complexity results, and algorithmic 
 techniques has exploded recently, as evident by the publication of three recent books on the 
 subject \cite{HusonBook,morrison2011introduction,gusfield2014recombinatorics}. A main premise behind the 
 use of phylogenetic networks is that when a single tree is not sufficient to model the evolutionary history of a 
 set of sequences or characters, a phylogenetic network that encompasses several trees is used. For example,
 the phylogenetic network in Fig. \ref{fig:net}(a) depicts an evolutionary history that involves hybridization between 
 taxon D and the most recent common ancestor (MRCA) of taxa B and C. 
\begin{figure}[!ht]
\begin{center}
\includegraphics[width=2.5in]{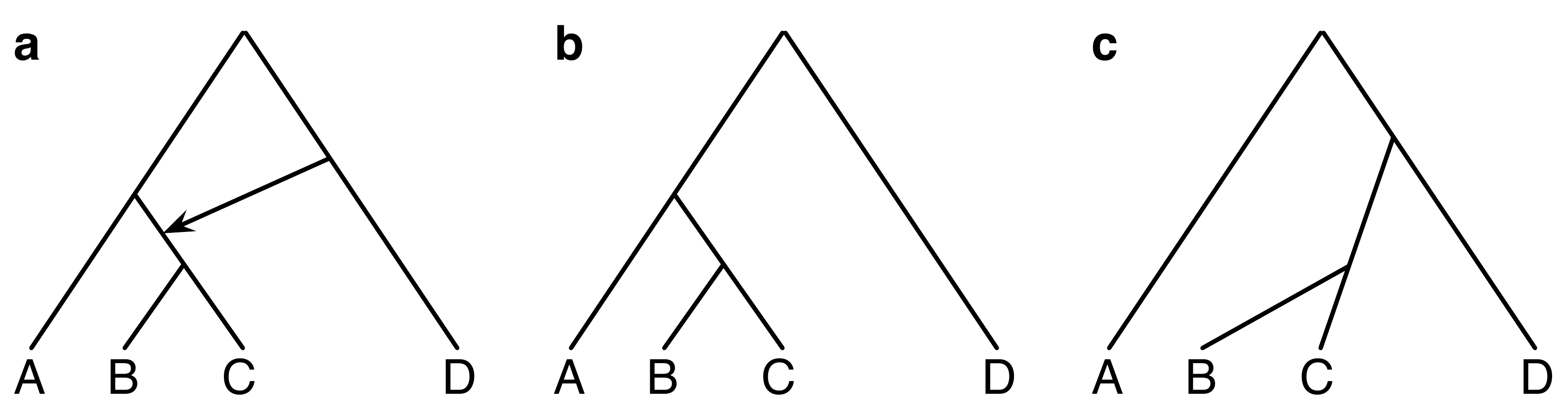} 
\vspace{-.15in}
\caption{A phylogenetic network $\psi$ with one reticulation node (a) and two trees it displays (b,c).\label{fig:net}}
\end{center}
\end{figure}

 Central to research on phylogenetic networks has been the notion of trees {\em displayed} by a phylogenetic network.
We say that a phylogenetic network displays a tree if the tree can be obtain be removing a set of ``reticulation edges" 
of the network. Fig. \ref{fig:net} shows the two trees displayed by the  network given in the figure. Given a phylogenetic 
network $\psi$, we denote by ${\cal U}(\psi)$ the set of all trees displayed by $\psi$. When incongruence in the gene 
trees inferred on different genomic regions across a genome alignment is assumed to be caused only by reticulation (e.g.,
hybridization), then the observed gene trees are taken to be a subset of the set of trees displayed by the (unknown) 
phylogenetic network for the set of genomes. This is why the set ${\cal U}(\psi)$ has played a fundamental role in most 
results established for phylogenetic networks.  Examples of the prominent use of ${\cal U}(\psi)$ include: (1) Parsimonious phylogenetic networks that fit the 
evolution of a sequence of sequences under the infinite sites model \cite{wang2001perfect,nakhleh2005perfect,gusfield2007decomposition,gusfield2003efficient,song2006algorithms,song2003parsimonious,song2004minimum,song2005constructing}; 
 (2) extending the maximum parsimony and maximum likelihood criteria from trees to networks \cite{hein90,NakhlehJin05,JinNakhleh-eccb,JinNakhleh-bioinfo,JinNakhleh-isbra,JinNakhleh-mbe}; (3) inferring minimal networks from sets of gene trees \cite{BaroniSemple06,HusonRupp08,van2010phylogenetic,wu2013algorithm}; (4) establishing identifiability results related to networks \cite{pardi2015reconstructible}; (5) establishing complexity results related to networks \cite{kanj2006reconstructing,bordewich07b,kanj2008seeing,kanj2008compatibility,van2010locating,van2011two}; and (6) identifying special 
 trees within the network \cite{steel2013identifying,daskalakis2015species,davidson2015phylogenomic,francis2015phylogenetic}.
	
 One of the evolutionary phenomena that has been extensively documented in recent analyses and targeted for 
 computational developments is {\em deep coalescence}, or {\em incomplete lineage sorting} \cite{rosenberg2002probability}. 	
 This phenomenon amounts to gene tree incongruence due to population effects (e.g., the size of an ancestral population or the 
 time between two divergence events). When this phenomenon is present in a reticulate evolutionary history, a major challenge 
 faces all the aforementioned works: The set of trees displayed by a network is no longer adequate to fully capture gene evolution 
 within the network. To resolve this issue, we define the set of parental trees of a phylogenetic network
 to supplant the set of displayed trees. Based on this set, we make three contributions. First, we extend the concept of anomaly zone 
 to phylogenetic networks and establish results based on this concept. It is important to note here that Sol{\'\i}s-Lemus {\em et al.} 
  \cite{solis2016inconsistency} recently discussed the issue of anomaly in the presence of reticulation where they focused on the ``species tree" inside the 
  network. Here, we define the anomaly zone in terms of the whole set of parental trees and do not designate a species tree inside the 
  network. Second, we address the problem of inferring a backbone tree 
 inside the network that could serve as a starting tree for network searches and/or provide information on a potential species tree 
 despite reticulation. As in the first contribution, the work here differs from that of  \cite{solis2016inconsistency} in focusing on all trees displayed by a network,
 rather than just a designate species tree. Third, we propose a novel clustering-based approach to phylogenetic network inference from gene trees by which 
 the gene trees are first clustered, parental trees are inferred from the clusters, and then the parental trees are combined into a 
 phylogenetic network. Gori {\em et al.} \cite{Gori01062016} recently studied the performance of various combinations of dissimilarity measures 
 and clustering techniques in clustering gene trees. Our work differs from that of \cite{Gori01062016} in that our focus is on phylogenetic network inference 
 via clustering.  We believe our work will open up new venues for research into computational methods and mathematical results 
 for reticulate evolutionary histories. 
\vspace{-.15in}
\section*{Methods}
We focus here on binary evolutionary (or, explicit) phylogenetic networks \cite{nakhleh2010evolutionary}.
\vspace{-.15in}
\begin{definition}
The topology of a phylogenetic network $\psi$ is a rooted directed acyclic graph $(V,E)$ such that $V$ contains 
a unique node with in-degree $0$ and out-degree $2$ (the root) and each of the other nodes has either in-degree $1$ and out-degree $2$ (an internal tree node),
an in-degree $1$ and out-degree $0$ (an external tree node, or leaf), or in-degree $2$ and out-degree $1$ (a reticulation node). The phylogenetic network 
has branch lengths $\lambda$, such that $\lambda_b$ denotes the length, in coalescent units (in coalescent units, the length of an edge in number of generations is divided by twice the size of the effective population size of 
 the population associated with that edge, and is a standard unit used in coalescent theory), of branch $b$ in $\psi$. 
\end{definition}
\vspace{-.15in}
 As we discussed in the Background section and illustrated in Fig. \ref{fig:net}, the notion of trees displayed by a network 
 has played a central role in analyzing and inferring networks. 
\vspace{-.15in}
 \begin{definition}
 Let $\psi$ be a phylogenetic network. A tree $t$ is displayed by $\psi$ if it can be obtained by removing for each reticulation node 
 one of the edges incident into it followed by repeatedly applying forced contractions until no nodes of in- and out-degree $1$ remain. 
 We denote by ${\cal U}(\psi)$ the set of all trees displayed by $\psi$. 
 \end{definition}
 \vspace{-.15in}
 Fig. \ref{fig:net} shows a phylogenetic network $\psi$ along with ${\cal U}(\psi)$. 
\vspace{-.15in}
\subsection*{Deep coalescence and the parental trees inside a network}
 Let us consider tracing the evolution of a recombination-free genomic region
 of four individuals $a$, $b$, $c$, and $d$, sampled from the four taxa A, B, C and D within the branches of the phylogenetic network 
 $\psi$ of Fig. \ref{fig:net}. If $b$ and $c$ coalesce at the most recent common ancestor (MRCA) of B and C, and no events such as 
 deep coalescence or duplication/loss occur anywhere in the phylogenetic network, then the genealogy of 
 the genomic region is one of the two trees in the set ${\cal U}(\psi)$. This is precisely the reason why much attention has been given 
 to the set ${\cal U}(\psi)$, as discussed in the Background section. 

 However, let us now consider a scenario where $b$ and $c$ did not coalesce at the MRCA of B and C. One potential outcome 
 in terms of the resulting genealogy for $a$, $b$, $c$, and $d$ is illustrated in Fig. \ref{fig:dc}(a). The probability that $b$ and $c$ 
 fail to coalesce at the MRCA of B and C has to do with the quantity $y$ in the figure: The smaller it is, the more likely it is that $b$
 and $c$ would fail to coalesce \cite{DegnanSalter05}. 
\begin{figure}[!ht]
\begin{center}
\includegraphics[width=2.2in]{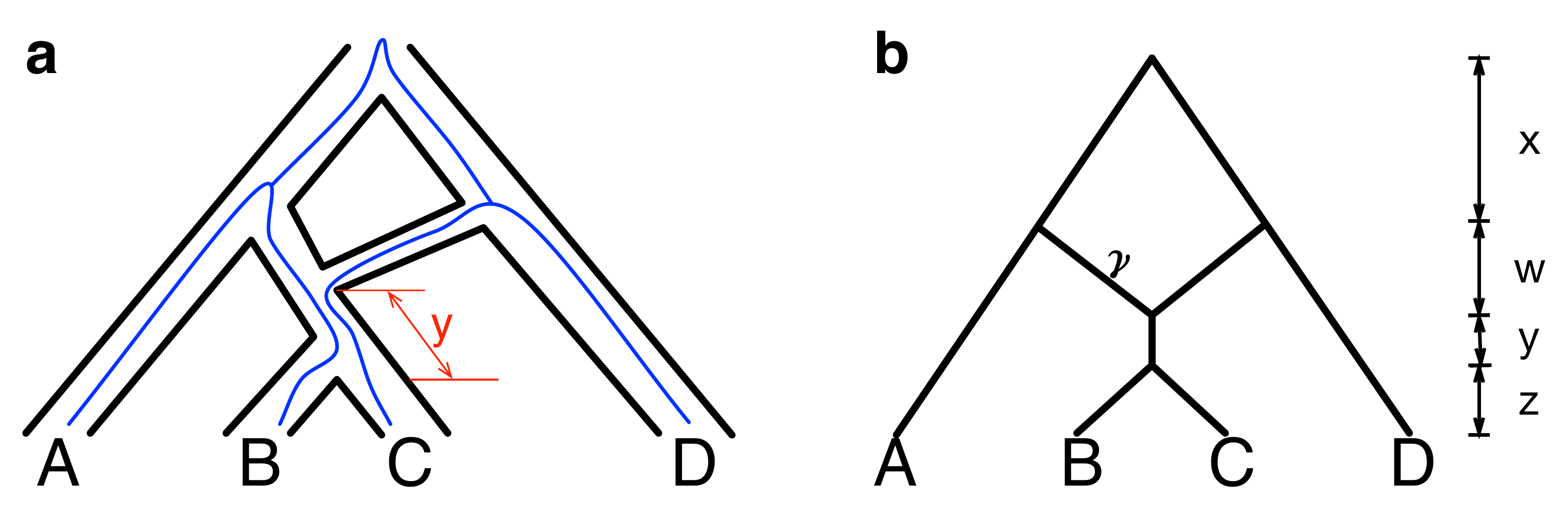}
\vspace{-.15in}
\caption{(a) The genealogy of a gene (blue lines) within the branches of a phylogenetic network. In this scenario, 
the two lineages from B and C failed to coalesce prior to the reticulation node (evolution is viewed backward in time,
from the leaves toward the root). The resulting genealogy in this case is ((a,b),(c,d)) and neither of the two trees 
in the set ${\cal U}(\psi)$ (shown in Fig. \ref{fig:net}) capture this scenario. The length in coalescent units of the branch between the reticulation
node and the MRCA of B and C is $y$. (b) An abstract representation of the network, assuming both reticulation edges 
have the same length $w$. \label{fig:dc}}
\end{center}
\end{figure}
Interestingly, for the scenario illustrated in Fig. \ref{fig:dc}(a), neither of the two trees in the set ${\cal U}(\psi)$ can capture 
the shown genealogy. This brings us to define the set of parental trees inside a phylogenetic network to appropriately 
represent the network as a mixture of trees that adequately model the evolution of genes in the presence of deep coalescence.

Yu {\em et al.} \cite{yu2012probability} gave an algorithm for the simple task of converting a phylogenetic network $\psi$ to a multi-labeled tree, or MUL-tree, $T$.
 Proceeding from the leaves of the network toward the root, the algorithm creates two copies of each subtree rooted at a reticulation node, attaches them to 
 the two parents of the reticulation node, and deletes the two reticulation edges. See Fig. \ref{fig:net2}(a) for an illustration. Notice that multiple leaves could be 
 labeled with the same taxon name, and hence the MUL-tree naming.  Due to page limitations, we provide the pseudo-code of the algorithm in the Appendix.  
 
 As phylogenomic analyses are increasingly involving multiple individuals per species, we provide a general definition that applies to cases 
 with multiple individuals per species. Let ${\cal X}$ be the set of species and $a_x$ denote the number of genomes sampled from species  $x \in {\cal X}$. 
 Let $T$ be a MUL-tree. We denote by $T|_{({\cal X},a)}$ a tree obtained from $T$ by retaining, for each taxon $x \in {\cal X}$, $a(x)$ or fewer leaves labeled by 
 $x$ and deleting the remaining leaves labeled by $x$, followed by repeatedly applying forced contractions until no nodes of in- and out-degree $1$ remain. 
  \vspace{-.15in}
\begin{definition}
 Let $\psi$ be a phylogenetic network on set ${\cal X}$ of taxa and $T$ be its MUL-tree. A parental tree inside $\psi$ is a tree $t$ such that $t=T|_{({\cal X},a)}$. 
 We denote by ${\cal W}(\psi)$ the set of all parental trees inside $\psi$. 
 \end{definition}
 \vspace{-.15in}
 Fig. \ref{fig:net2} shows the set ${\cal W}(\psi)$ for the phylogenetic network in Fig. \ref{fig:net}. 
\begin{figure}[!ht]
\begin{center}
\includegraphics[width=4.5in]{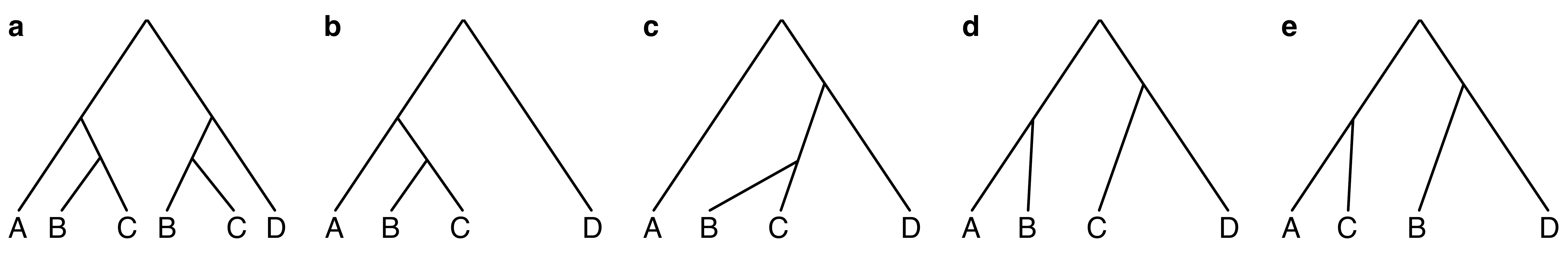}
\vspace{-.15in}
\caption{(a) The MUL-tree of the phylogenetic network $\psi$ in Fig. \ref{fig:net}(a). The set ${\cal W}(\psi)$ consists of the four trees in (b)--(e), assuming 
 one individual is sampled per species. \label{fig:net2}}
\end{center}
\end{figure}
 The genealogy shown in Fig. \ref{fig:dc}(a) can be captured by the parental tree in Fig. \ref{fig:net2}(d). Indeed, Yu {\em et al.} 
 \cite{yu2012probability,yu2014maximum} gave mass and density functions for gene trees on phylogenetic networks in terms of the 
 set of parental trees inside the network. While it is obvious that ${\cal U}(\psi) \subseteq {\cal W}(\psi)$, the two sets can differ significantly 
 in terms of their properties. For example, if $\psi$ has $k$ reticulation nodes, then $|{\cal U}(\psi)| \leq 2^k$. However, $|{\cal W}(\psi)|$ could 
 be much larger than $2^k$, as it is a function of the numbers of leaves under the reticulation nodes as well as the numbers of individuals sampled 
 per species. 
\vspace{-.15in}
\subsection*{Inheritance probabilities and the multispecies network coalescent}
 Given a species tree topology $\psi$ and its branch lengths $\lambda$, the gene tree topology $G$ can be viewed as 
 a discrete random variable whose mass function  $P_{\psi,\lambda}(G=g)$ was derived in \cite{DegnanSalter05}. In the case of phylogenetic
 networks, we also associate with every pair of edges $b_1=(u_1,v)$ and $b_2=(u_2,v)$  that are incident into the same reticulation node $v$ nonnegative  
 real values $\gamma_{b_1}$ and $\gamma_{b_2}$ such that $\gamma_{b_1}+\gamma_{b_2}=1$ \cite{yu2012probability,yu2014maximum}. These quantities, which we call inheritance probabilities, 
  indicate the proportions of lineages in hybrid populations that trace each of the two parents of that population. In this case, the phylogenetic network's 
  topology $\psi$ and branch lengths $\lambda$, along with the inheritance probabilities $\Gamma$, are sufficient to 
  describe the mass function of gene trees $P_{\psi,\lambda,\Gamma}(G=g)$ under the multispecies network coalescent \cite{yu2012probability,yu2014maximum}. 
  \vspace{-.15in}
\section*{Results and discussion}
In this section we describe the three main contributions of this work. First, we extend the concept of anomaly zones \cite{degnan2006discordance} to phylogenetic networks 
and establish conditions for their existence. Second, we address the question of whether it is possible, from an inference perspective, to obtain a tree 
that can be augmented into the correct network by adding reticulation edges between pairs of the tree's edges. Third, we propose a clustering approach 
to network inference by clustering the gene trees, inferring parental trees, and then combining the parental trees into a network. These results have 
direct implications not only on understanding the relationships between trees and networks, but also the practical task of developing computational methods 
for network inference. 
\vspace{-.15in}
\subsection*{Phylogenetic networks and anomalies}
 In a seminal paper, Degnan and Rosenberg \cite{degnan2006discordance} showed that the branch lengths of a species tree could be set 
 such that the most likely gene tree disagrees with the species tree. Such a gene tree is called an {\em anomalous gene 
 tree} and the set of all branch length settings that result in an anomalous gene tree is the {\em anomaly zone}.  
  
 We now provide what, to the best of our knowledge, is the first definition of anomaly zones for phylogenetic networks. Note that in \cite{solis2016inconsistency},
 Sol{\'\i}s-Lemus {\em et al.} discussed anomalous gene trees in the presence of ILS and gene flow. However, in their work, the anomaly was still defined with 
 respect to a designated species tree (they viewed the phylogenetic network as a species tree with additional horizontal edge between pairs of its branches). Here, 
 we do not designate any of the parental trees of the network as a species tree; instead, we define the anomaly zone directly in terms of the entire set. 
 
  The guiding principle 
 behind our definition is the question: Is the most likely gene tree to be generated by a phylogenetic network necessarily a parental tree inside 
 the network? 
 \vspace{-.15in}
 \begin{definition}
 Let $\psi$ be a phylogenetic network, $\lambda$ be its branch lengths, and $\Gamma$ be the inheritance probabilities associated 
 with its reticulation edges. 
We say gene tree topology $g$ is  anomalous for $(\psi,\lambda,\Gamma)$ if
\begin{equation}
\label{anom}
P_{\psi,\lambda,\Gamma}(G=g)>P_{\psi,\lambda,\Gamma}(G=t) \ \ \ \forall{t \in {\cal W}}(\psi).
\end{equation}
A phylogenetic network $\psi$ is said to produce anomalies if there exists branch lengths $\lambda$ and inheritance probabilities $\Gamma$
such that there exists an anomalous gene tree $g$ for $(\psi,\lambda,\Gamma)$.
The anomaly zone for a phylogenetic network $\psi$ is a set of  $(\Lambda,\Gamma)$ values for which $\psi$ produces anomalies.
\end{definition}
\vspace{-.15in}
Degnan and Rosenberg \cite{degnan2006discordance} showed 
 that three-taxon and symmetric four-taxon species trees have no anomaly zones, but that non-symmetric four-taxon trees and all species trees 
 with five or more taxa have anomaly zones. One practical implication of these results was that the simple approach of sampling a very large 
 number of loci, building gene trees and taking the most frequent gene tree as the species tree (an approach dubbed ``the democratic vote" method) 
 does not always work. 
 
 Since the multispecies coalescent is a special case of the multispecies network coalescent, it immediately follows that 
 any phylogenetic network with $n \geq 5$ leaves produces anomalies. We now show that three-taxon phylogenetic 
 networks do not produce anomalies, but that symmetric phylogenetic networks with $n=4$ leaves could produce anomalies. 
 Note that according to \cite{solis2016inconsistency}, 3-taxon networks could still generate anomalous gene trees. The seeming 
 discrepancy between the two results is due to to the fact that here we define the anomaly zone in terms of all the parental trees inside 
 the network and not just one designate species tree. 
 
\vspace{-.15in}
\begin{lemma} 
\label{lemm1}
A phylogenetic network $\psi$ on 3 taxa does not produce anomalies.
\end{lemma}
\vspace{-.15in}
(Proof is in the Appendix.) Consider now the symmetric phylogenetic network $\psi$ in Fig. \ref{fig:dc}(b) and whose set of parental trees in 
given in Fig. \ref{fig:net2}. 
%\begin{figure}[!ht]
%  \centering
%  \includegraphics[width=0.2\textwidth]{figures/4TaxaSpecial0}
%  \vspace{-.15in}
%  \caption{A symmetric phylogenetic network with 4 taxa and 1 reticulation node. The two horizontal edges have length $0$ each. 
%  \label{fig:4taxa}}
%\end{figure}
 The four gene trees that are identical to the parental trees of the network are  $((a,(b,c)),d)$, $(a,((b,c),d))$, $((a,b),(c,d))$ and $((a,c),(b,d))$. 
 We plotted the most probable gene tree of this network when $x$ and $y$ are both small in Fig. \ref{fig:4TaxaSpecialProb}, in which yellow region and orange region stand for the anomaly zone of this network, and blue region is where the most 
  probable gene tree be the backbone tree. This figure shows the existence of anomaly zone of the network in Fig. \ref{fig:dc}(b) (where $w$ is set to $0$), which means that symmetric phylogenetic networks with $n=4$ leaves could produce anomalies.

 \begin{figure}[!ht]
  \centering
  %  \subfigure[$\gamma=0.5$]{
  %	  \label{fig:4TaxaSpecialProb:a}
      \includegraphics[width=0.4\textwidth]{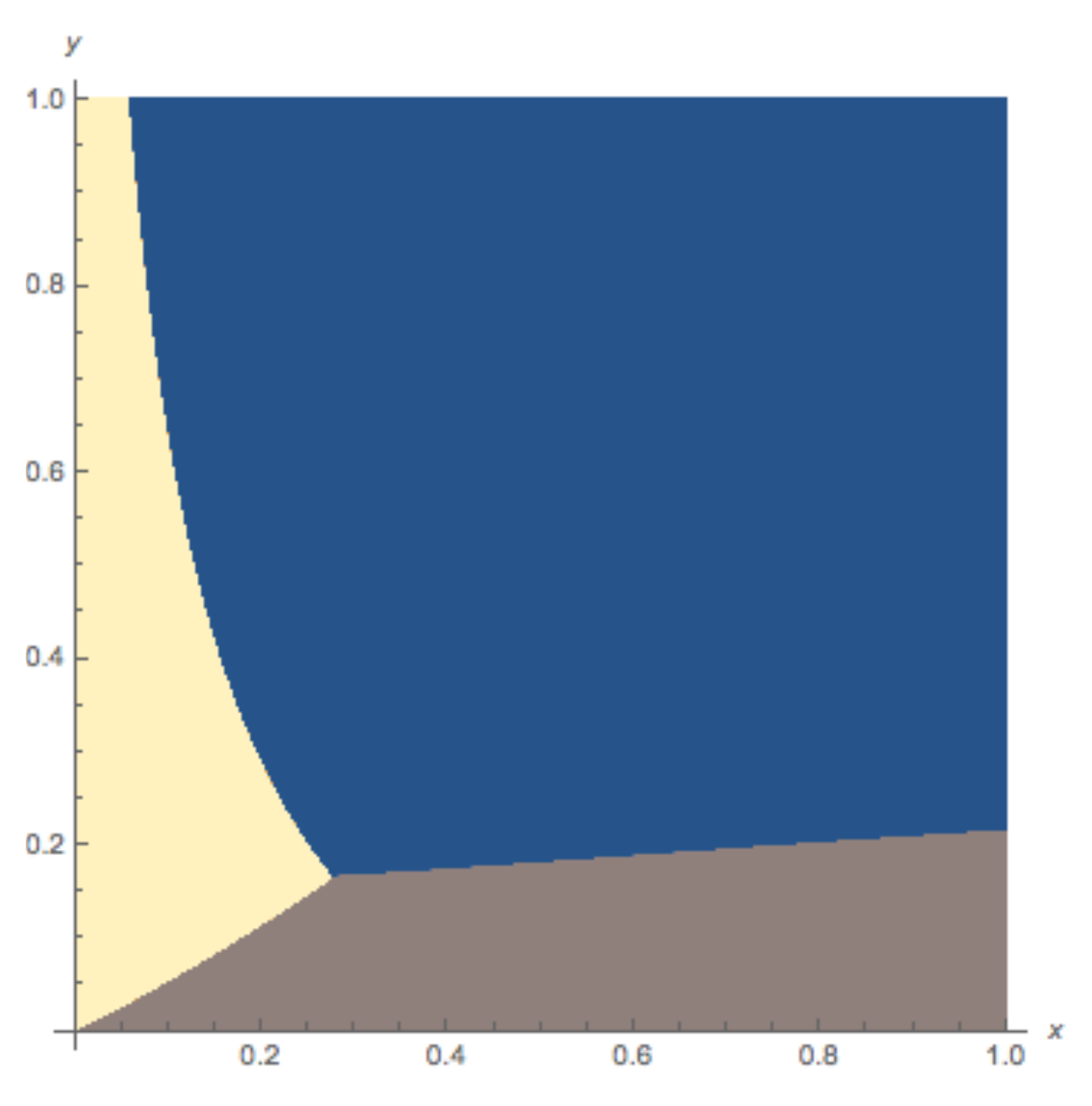}
%    \subfigure[$\gamma=0.05$]{
  %	  \label{fig:4TaxaSpecialProb:b}
      \includegraphics[width=0.4\textwidth]{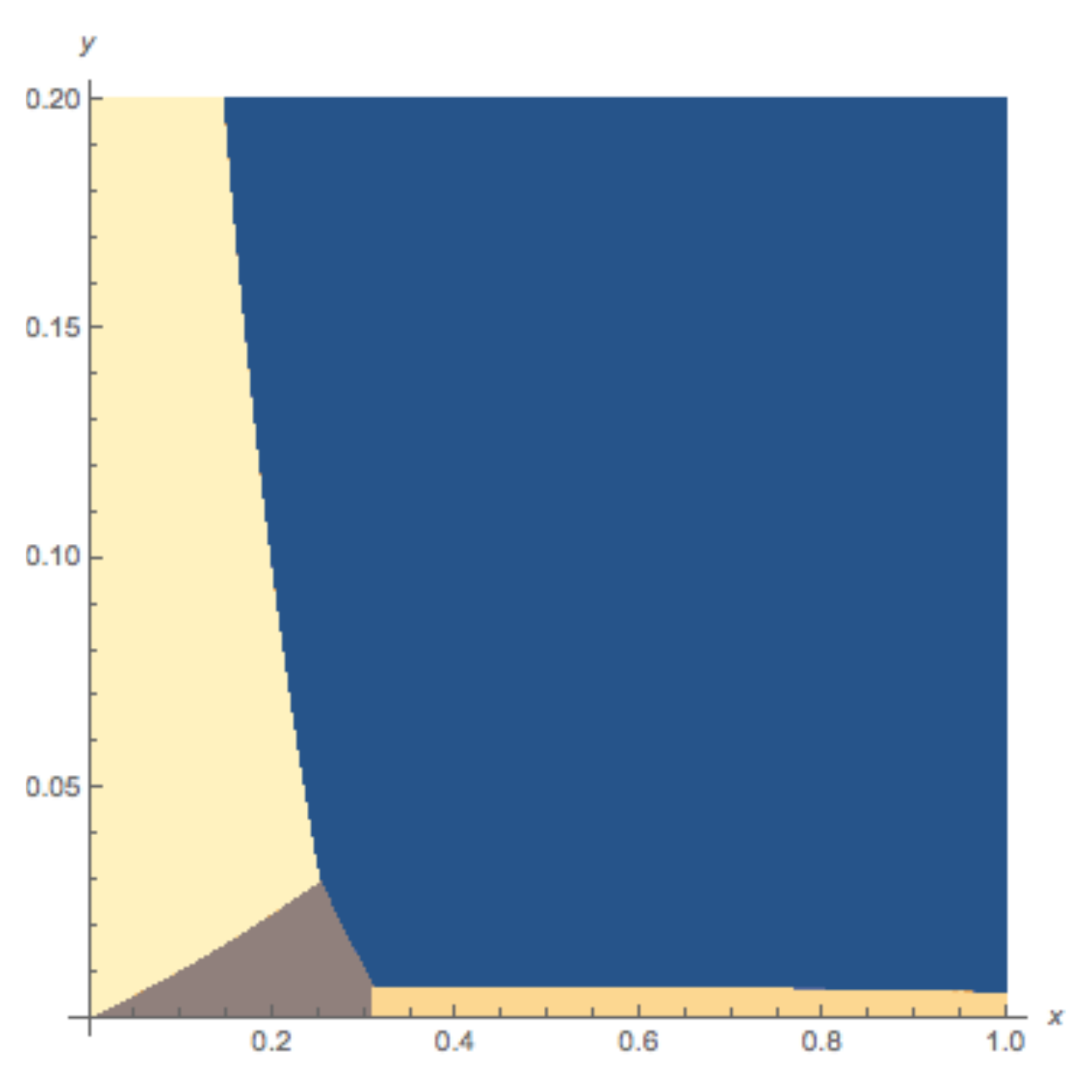}
      \vspace{-.15in}
  \caption{The most likely gene tree give the phylogenetic network in Figure \ref{fig:dc}(b) (with $w=0$) with $\gamma=0.5$ (left) and $\gamma=0.05$ (right). The y-axis is branch length $y$. The x-axis is branch length $x$. Blue: gene tree $(a,((b,c),d))$ in both panels, and gene tree $((a,(b,c)),d)$ additionally in the left panel; Yellow: gene tree $((a,d),(b,c))$;
  Orange: gene trees $(a,(b,(c,d)))$ and $(a,(c,(b,d)))$; Brown: gene trees $((a,b),(c,d))$ and $((a,c),(b,d))$. 
  \label{fig:4TaxaSpecialProb}}
\end{figure}

 \subsection*{On the backbone tree of a phylogenetic network}
 A very important question in the area of phylogenetic network inference is whether there exists a tree that 
 can be augmented into the network by adding reticulation edges between pairs of the tree's edges. Here, we refer 
 to such a tree as the network's {\em backbone tree}. A biological significance of this tree lies in its potential designation
 as the species tree (e.g., see the species tree underlying the phylogenetic network of mosquitos in \cite{FontaineEtAl15}). 
 
 Francis and Steel \cite{francis2015phylogenetic} recently introduced the notion of {\em tree-based networks} to capture 
 those networks that can be obtained by augmenting a backbone tree (they called it the ``base tree"). We now show that even
 if a network is tree-based, it is not necessarily the case that the most likely gene tree is its base, a result that is related to the 
 anomalous gene trees discussed above. Let us consider again the network of Fig. \ref{fig:dc}(b). This network is tree-based 
 and each of the two trees in Fig. \ref{fig:net} could serve as its base (indeed, the same network is drawn in Fig. \ref{fig:net} in 
 a way that clearly demonstrates that it is tree-based). The probabilities of all 15 gene trees under this phylogenetic network are 
 given in Table \ref{tab:table1} in the Appendix. 
 While there are 15 possible gene tree topologies on taxa $a$, $b$, $c$, and $d$, as branch length $x$ in the network tends to 
 infinity, the probability of seven of the 15 gene tree topologies converges to $0$ and only eight gene trees have non-zero mass: 
 $((a,(b,c)),d)$, $(a,((b,c),d))$, $((a,b),(c,d))$, $((a,c),(b,d))$, $(((a,b),c),d)$, $(((a,c),b),d)$, $(a,(b,(c,d)))$, and $(a,(c,(b,d)))$. The probabilities in this 
 case are given in Table \ref{tab:table1x} in the Appendix and visualized as a function of varying branch length $y$ for two different 
 settings of $\gamma$ in Fig. \ref{fig:4taxa_chart} in the Appendix. When $\gamma=0.5$ and $\frac{1}{4}e^{-y}>\frac{1}{2}-\frac{5}{12}e^{-y}$, which is equivalent 
 to  $y<0.288$, the most likely gene tree given $\psi$ is not one of its base trees (that is, the network cannot be obtained by a adding 
 a single reticulation edge to the most likely gene tree). This also demonstrates that if we defined anomalies in terms of the set 
 ${\cal U}(\psi)$ instead of set ${\cal W}(\psi)$, the phylogenetic network would still produce anomalous gene trees. 
 
Given that the most likely gene tree is not necessarily a backbone of the phylogenetic network, we now turn our attention to three 
recent methods whose goal is to infer a species tree despite horizontal gene transfer. It is very important to point out upfront that the 
assumptions of these methods do not necessarily match the scenarios we investigate here, but our goal is to assess how well they 
do at recovering a backbone tree inside the network of Fig. \ref{fig:dc}(b). In \cite{davidson2015phylogenomic}, Davidson {\em et al.}  showed that
  ASTRAL-II \cite{mirarab2015astral} performed best among species tree inference methods in terms of recovering the species tree in the 
  presence of reticulation (under a specific model of horizontal gene transfer). They further proved that the method is statistically consistent 
  in terms of recovering the species tree under the same model. 
  In \cite{steel2013identifying}, Steel {\em et al.} showed that triplet-based approaches to species tree inference are 
 consistent in terms of inferring a species tree in the presence of horizontal gene transfer (also under a specific model).  This technique was 
 implemented as the ``primordial tree" in Dendroscope \cite{huson2012dendroscope}. Both ASTRAL-II and the primordial tree method in 
 Dendroscope take 
 gene trees as input. The method of Daskalakis and Roch \cite{daskalakis2015species} takes as input   
 gene trees with branch length and compute the distance between every two taxa $u$ and $v$ as the median 
 of the gene-tree distances between $u$ and $v$ over all gene trees in the data set (given a gene tree with branch lengths, the gene-tree 
 distance between two leaves is the sum of the branch lengths on the simple path between the two leaves). 
 
%  ms 4 2000 -T -I 4 1 1 1 1 -ej 500.0 3 2 -es 500.05 2 0.5 -ej 500.05 2 4 -ej 500.05 5 1 -ej 1000.05 4 1
 We simulated gene tree data sets under the phylogenetic network of Fig. \ref{fig:dc}(b) using ms \cite{hudson2002generating} while 
 varying branch length $y$ to take on values from the set $\{0.1,0.2,0.5,1.0\}$ ($w$ was set to $0$ and $x$ was set to $1000$ so as to rule out deep coalescence
 involving the two branches incident with the root). Data sets with $25$, $50$, $100$ and $200$ gene trees were generated, 
 and for each configuration of branch length $y$ and number of gene trees, 100 data sets were simulated. 
 The accuracy of each method for a setting of branch length $y$ and number of gene trees is the fraction, out of the 100 data sets, 
 of times that the method returned one of the two trees displayed by the network. The results for all three methods on the simulated data 
 are shown in Fig. \ref{methods:perf}. 
 \begin{figure}[!ht]
\begin{center}
\begin{tabular}{cc}
\includegraphics[width=0.4\textwidth]{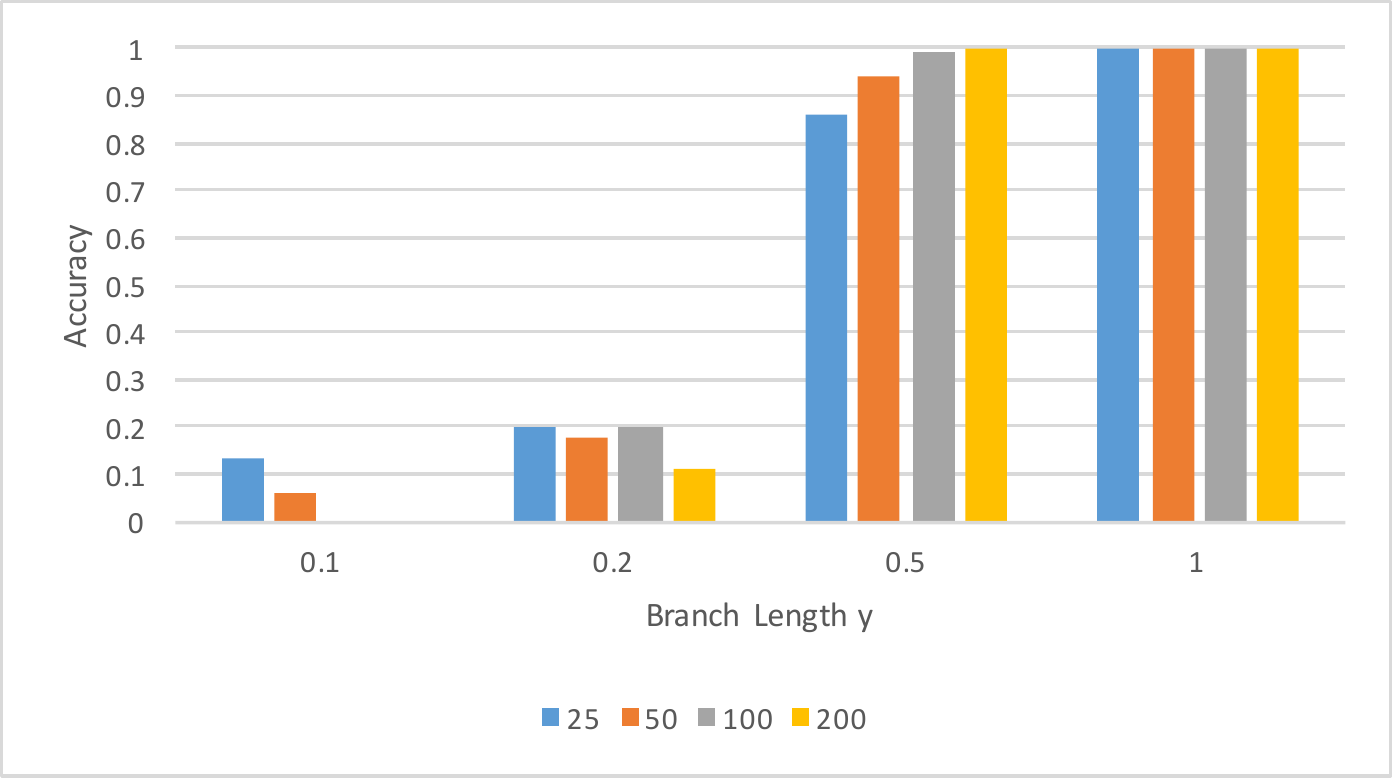} &  \includegraphics[width=0.4\textwidth]{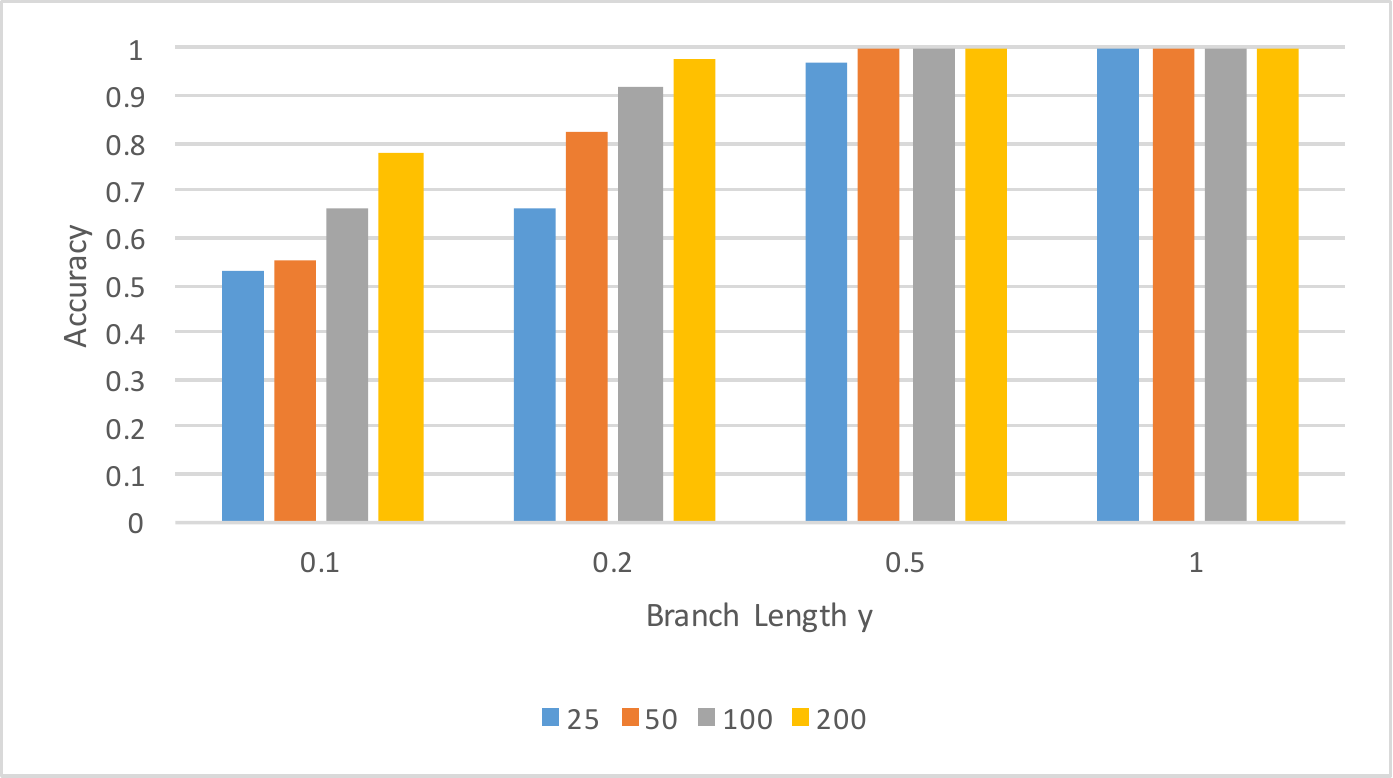} \\
 \includegraphics[width=0.4\textwidth]{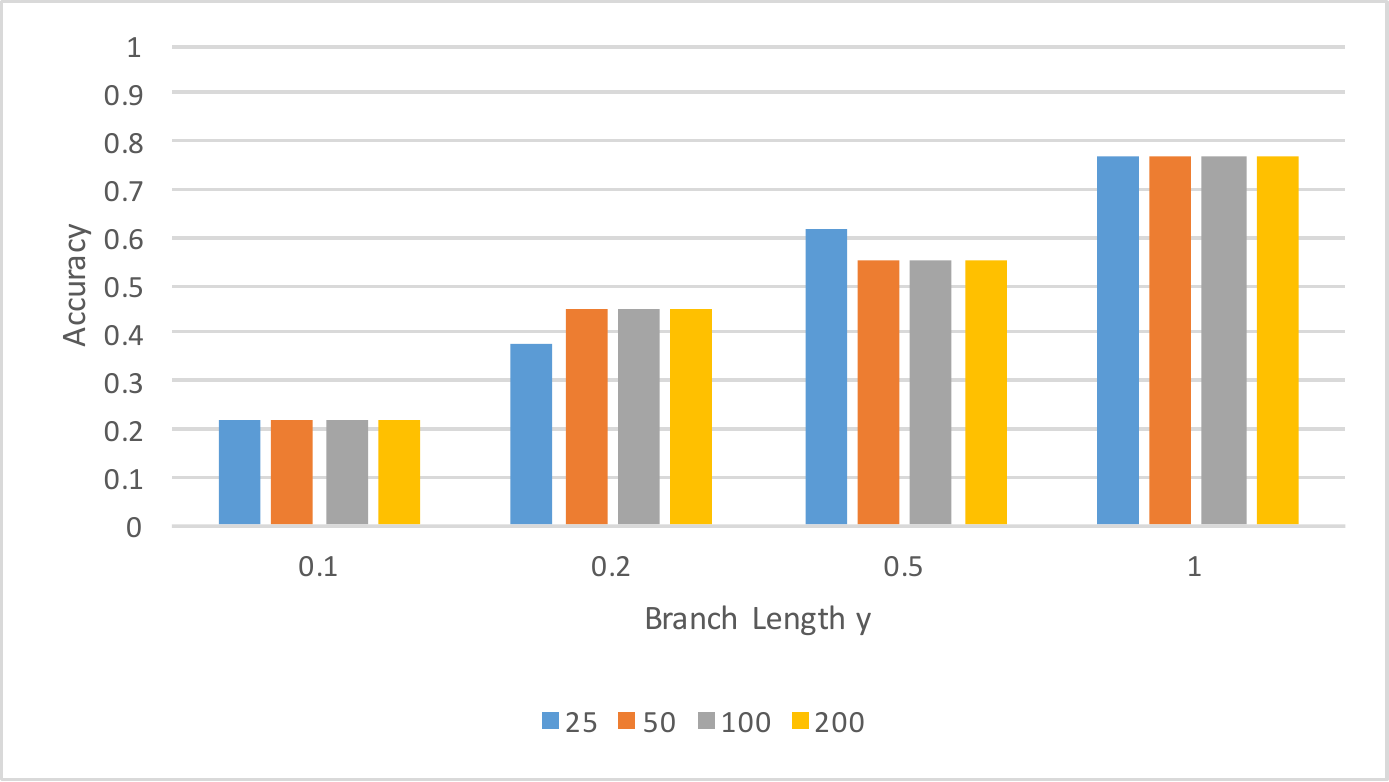} & \includegraphics[width=0.4\textwidth]{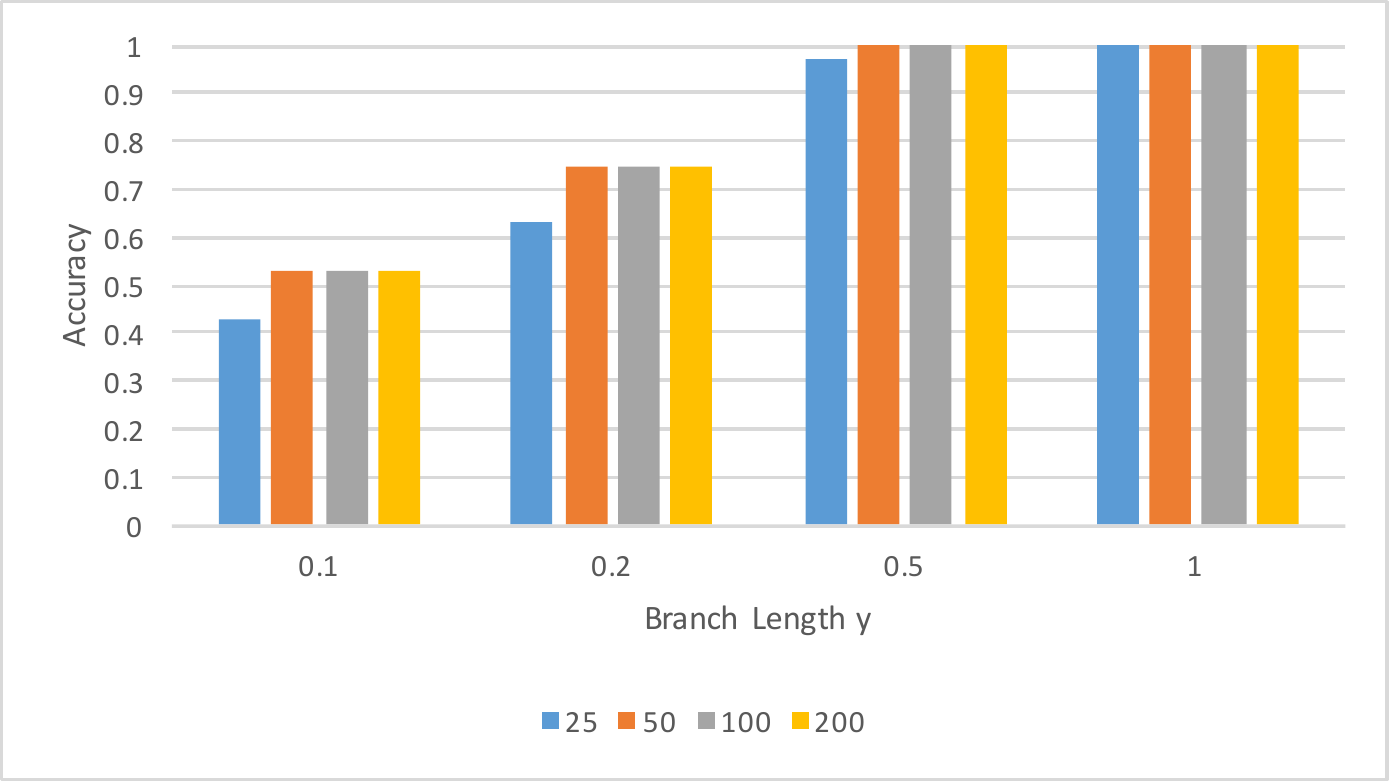}\\
 \includegraphics[width=0.4\textwidth]{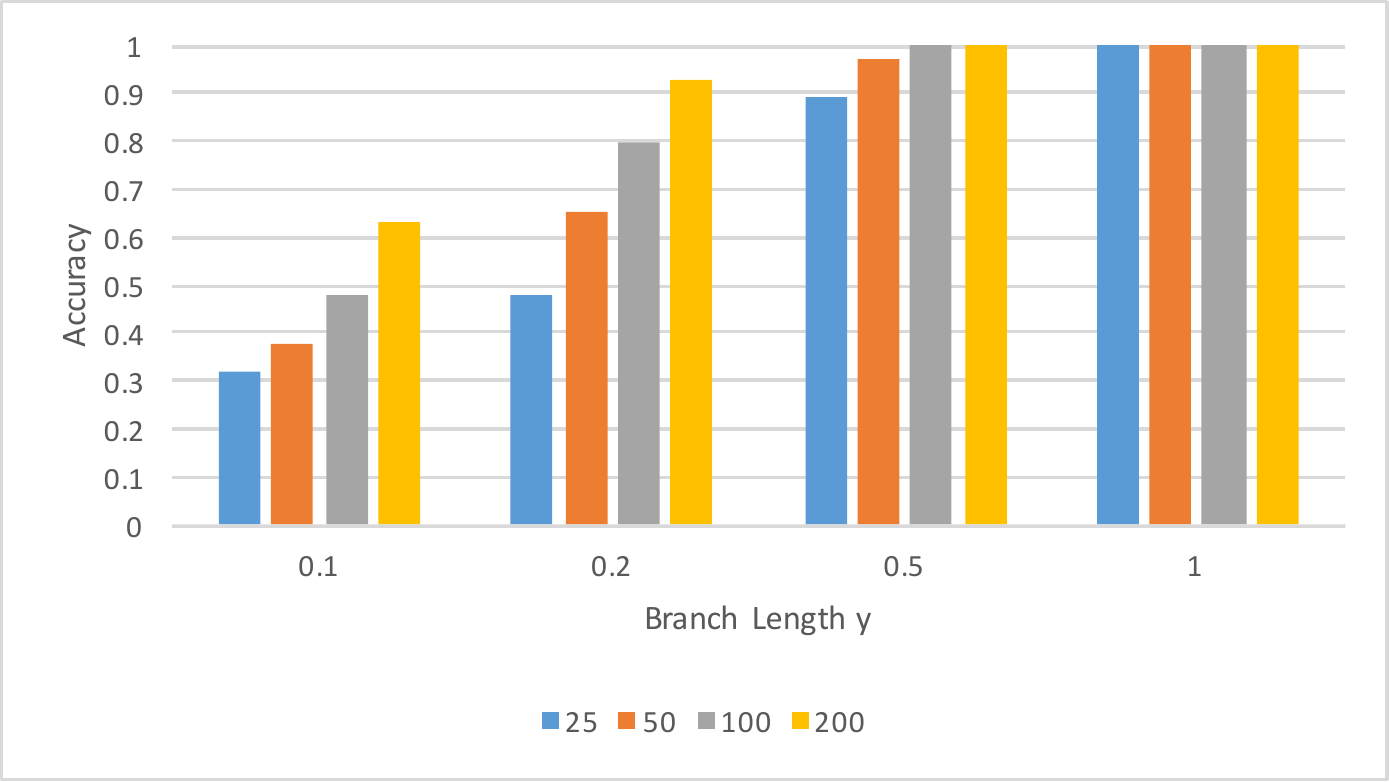} & \includegraphics[width=0.4\textwidth]{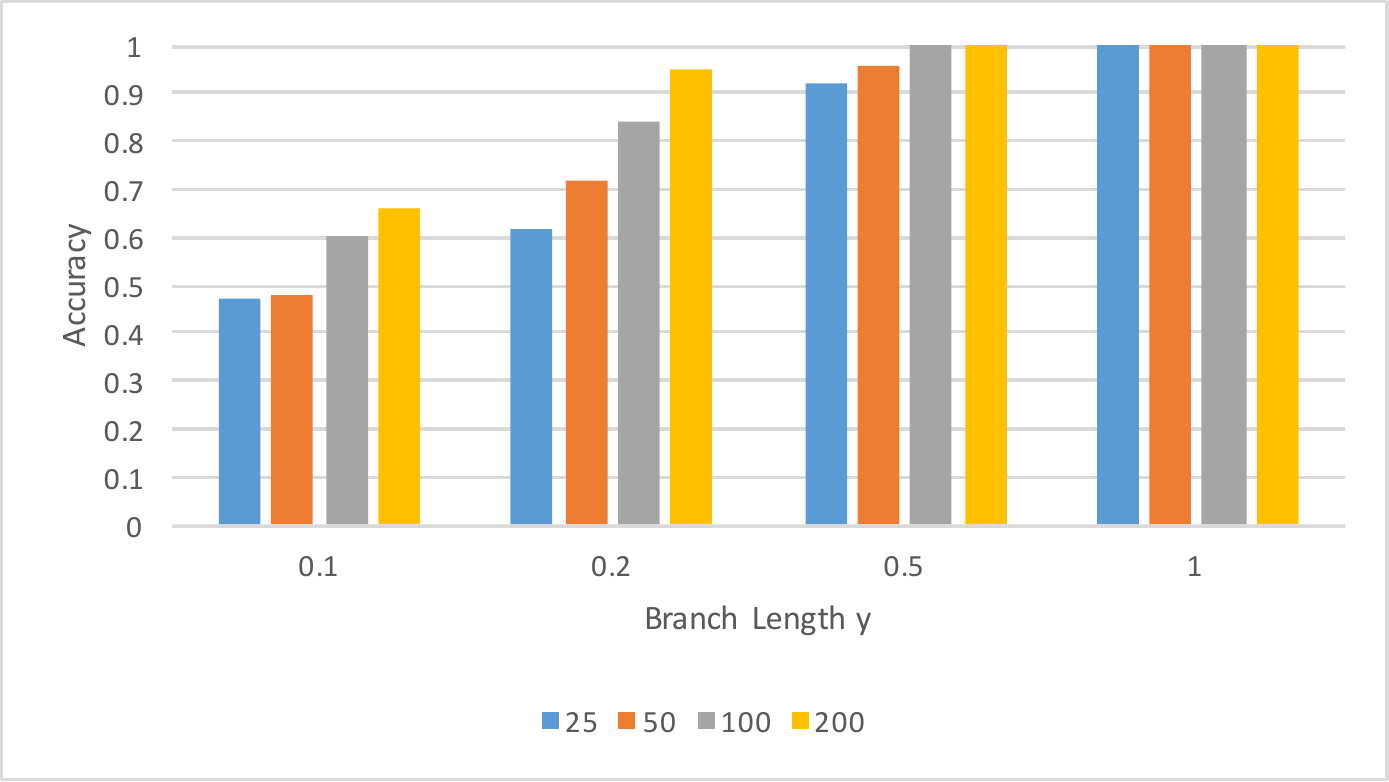} \\
\end{tabular}
\caption{The accuracy of three methods for inferring species trees in the presence of reticulation on data generated on the 
 phylogenetic network of Fig. \ref{fig:dc}(b). Left column corresponds to setting $\gamma=0.5$ and right column corresponds to setting $\gamma=0.05$. 
 Four settings for branch length $y$ were used, and for each setting data sets with 25, 50, 100, and 200 loci were generated. See the text for definition 
 of the accuracy measure. 
 (Top) ASTRAL-II \cite{mirarab2015astral}; (Middle) The 
 method of Steel {\em et al.}  \cite{steel2013identifying} as implemented in Dendroscope \cite{huson2012dendroscope}; (Bottom) Our own 
 implementation of the method of Daskalakis and Roch \cite{daskalakis2015species}.\label{methods:perf} }
\end{center}
\end{figure}

The results show that when $y$ is very small, the methods perform poorly in terms of returning one of the two trees displayed by the network,
 especially in the case of $\gamma=0.5$. This is expected as an inheritance probability of $0.5$ is a huge deviation from the assumptions of the 
 three methods. When $\gamma=0.5$ and $y$ is long enough (e.g., $1$), ASTRAL-II  and  the method of \cite{daskalakis2015species}
 do a perfect job, while the method of \cite{steel2013identifying} does not perform as well. For smaller values of $y$ and with $\gamma=0.5$, 
 the method of \cite{daskalakis2015species} consistently performs better than the other two methods. For $\gamma=0.05$, which is closer to the 
 assumptions of the  methods, all three of them perform well, even when $y=0.5$ (in this case, the most likely gene tree is also a backbone tree). For smaller values of $y$ in this case, ASTRAL-II and the method of \cite{daskalakis2015species} do almost equally well, and slightly better than the method of \cite{steel2013identifying}. 
 
 \vspace{-.15in}
\subsection*{From gene trees to species networks via parental trees: A clustering approach}
 Given our discussion above of the set of parental trees, one can view a phylogenetic network 
 $\psi$ as a mixture model with $|{\cal W}(\psi)|$ components and each component is a distribution 
 on gene trees defined by the parental tree corresponding to that component. 
  This view gives rise to a novel approach for reconstructing phylogenetic networks from a set ${\cal G}$ 
 of gene trees:
 	\begin{enumerate}
	\itemsep -1pt
	\item Cluster the gene trees into clusters $C_1,C_2,\ldots,C_k$;
	\item Infer a parental tree $T_i$ for cluster $C_i$ under the multispecies coalescent; 
	\item Combine the trees $T_1,T_2,\ldots,T_k$ into a phylogenetic network $\psi$. 
	\end{enumerate}
 The rationale behind this approach is that clustering would identify the components of the mixture model, 
 where the gene trees belonging to a component differ only because of incomplete lineage sorting (ILS), but not 
 because of hybridization. That is why in Step (2) a tree is inferred for each component under the multispecies 
 coalescent, which only handles ILS. In the third step, disagreements among the $k$ trees is assumed to be 
 all due to the hybridization events, and are used to obtain the final network. 
 A parsimony approach to Step (3) would be formulated as follows. 
 %which we 
 %pose as an open problem for future research. 
 \vspace{-.15in}
 \begin{definition} The Parental Tree Network Problem is defined as:
 \begin{itemize}
 \itemsep -1pt
 \item[] {\bf Input:} A set ${\cal P}$ of parental trees. 
 \item[] {\bf Output:} A phylogenetic network $\psi$ with the smallest number of reticulation nodes such that ${\cal P} \subseteq {\cal W}(\psi)$. 
 \end{itemize}
 \end{definition}
  \vspace{-.15in}
  
  In \cite{Gori01062016}, Gori {\em et al.} studied the performance of various combinations of clustering methods and dissimilarity measures on 
  gene tree topologies as well as gene trees with branch lengths. In our work here, the focus is on phylogenetic network inference and our 
  simulation study in what follows is preliminary and aimed at demonstrating the viability of this approach. 
  
 We used 10 phylogenetic networks (Fig. \ref{fig:clustering_result}(a)) to generate within each gene tree 
 data sets (50, 250, 500, and 1000 gene trees per data set and 30 data sets per configuration). 
% \begin{figure}[!ht]
%  \centering
%  \includegraphics[width=0.2\textwidth]{figures/Clustering1}
%   \vspace{-.15in}
%  \caption{The phylogenetic network used for the clustering results. The lengths of the two horizontal edges set to $0$ 
% and the indicated edge length to $0.2$ along with the specified inheritance probability of $0.35$. Ten networks were generated by setting the length of each other internal branch to a random number uniformly sampled in the range $[0.7,1.3]$. 
%  \label{fig:clustering}}
%\end{figure} 
\begin{figure}[!ht]
  \centering
%\begin{tabular}{cccc} 
 \includegraphics[width=0.9\textwidth]{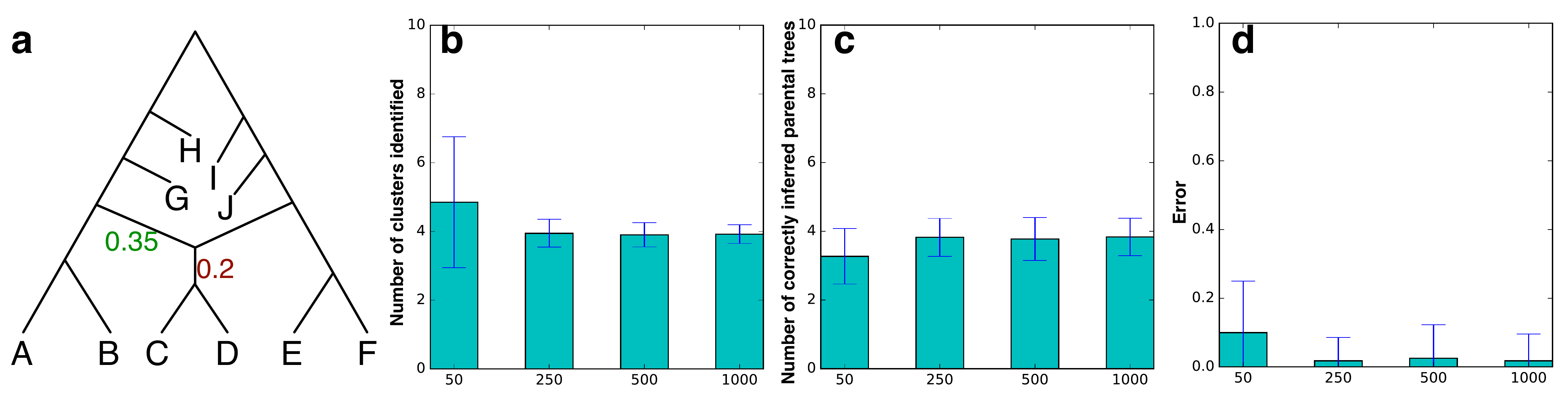} 
%& \includegraphics[width=0.26\textwidth]{figures/clustering_result1.pdf} \hspace{-.3in}
% &  \includegraphics[width=0.26\textwidth]{figures/clustering_result2.pdf} \hspace{-.3in}
% &  \includegraphics[width=0.26\textwidth]{figures/clustering_result3.pdf} \\ 
%\end{tabular}
 \vspace{-.15in}
  \caption{Performance of the clustering approach on the simulated data as a function of the number of 
  gene trees. (a) The phylogenetic network used in the simulations. The lengths of the two horizontal edges set to $0$ 
 and the indicated edge length to $0.2$ along with the specified inheritance probability of $0.35$. Ten networks were generated by setting the length of each other internal branch to a random number uniformly sampled in the range $[0.7,1.3]$. 
   (b) The number of clusters identified (averaged over 300 data sets for each bar).  
   (c) The number of correctly inferred parental trees (out of the maximum of four parental trees). 
   (d) The error between the set of inferred trees from the identified clusters and the set of four parental trees of the network. 
   The x-axis in panels (b)-(d) corresponds to the number of gene trees. 
  \label{fig:clustering_result}}
\end{figure}
  For each gene tree data set, pairwise Robinson-Foulds (RF) \cite{RF} distances were computed between the gene trees, and the pairwise distances were converted 
 into 3-dimensional points in Euclidean space using multidimensional scaling (MDS) as implemented in the  MDSJ package \cite{mdsj} (we also conducted clustering 
 directly on the RF distances, and found a significant improvement in the results after applying MDS). We implemented the $k$-means 
  clustering algorithm \cite{macqueen1967some} and used it to cluster the gene trees based on the Euclidean distances from MDS using $k=2,3,\ldots,10$. 
  We  implemented the silhouette method \cite{rousseeuw1987silhouettes} and the number of clusters with 
 the maximum average silhouette (based on the pairwise RF distances)  was selected as the number of clusters identified and the corresponding clustering as the identified clusters. 
  
  Fig. \ref{fig:clustering_result}(b) shows the results of identifying the number of clusters (the correct number is $4$). As the figure shows, 
  clustering in this case is performing very well, returning the correct number of clusters in almost all cases with 250 gene trees or more, and performing 
  slightly poorer in the case of 50 gene trees. 
  
After the clusters were identified, we turned to the next natural question: Do the clusters correspond to the parental trees of the 
network? To investigate this question, we chose to apply the ``minimizing deep coalescence" (MDC) method of \cite{ThanNakhleh-PLoSCB09} as 
implemented in \cite{ThanEtAl08} (the heuristic version that uses only the clusters in the input gene trees) to infer a ``species tree" on each cluster. 
We then quantified the number of true parental trees that were inferred by MDC on the clusters in each data set. The results are shown in Fig. \ref{fig:clustering_result}(c). The results indicate a very good performance where all four true parental trees are almost always correctly 
inferred, particularly when 250 gene trees or more are used. 
  
  Finally, when this MDC-based analysis returns trees other than the true parental trees, how far as they from the true ones? To answer this 
  question we compared the the set of true parental trees and the set of trees inferred by MDC based on the identified clusters using the 
  tree-based measure of \cite{warnow2003towards} (finding the min-weight edge cover of a bipartite graph whose two sets of nodes correspond to these two sets of trees 
  and the weights of edges are RF distance) as implemented in PhyloNet \cite{ThanEtAl08}. The results are shown in Fig. \ref{fig:clustering_result}(d).
  The results indicate a very good performance of about 2\% error for data sets with 250 gene trees or more, and about 10\% for data sets 
  with 50 gene trees. 
  
  It is worth mentioning that if a network that displays all gene trees in the input was sought, the result would be a network that differs significantly 
  from the true network, as each data set contained many distinct gene tree topologies. This highlights the major difference between the current 
  practice of seeking a network
  that displays all gene trees in the input and our proposed approach of seeking a network whose parental trees are obtained from the input 
  gene trees.

\vspace{-.15in}
\section*{Conclusions}
 In this paper, we showed that when deep coalescence occurs, inference and analysis of phylogenetic networks are more adequately 
 done with respect to the set of parental trees of the network, rather than the common practice of using the set of trees displayed by the network.
 We described the simple procedure for enumerating the set of parental trees of a given network, and based on this set, we made three contributions. 
 First, we defined the anomaly zone for a phylogenetic network topology as the region of branch lengths and inheritance probabilities under which 
 the most likely gene tree is not one of the parental trees inside the network. We provided straightforward results on the anomaly zones for networks
 that mainly result from the fact that networks are an extension of trees. An important question is whether it is possible that none of the trees displayed 
 by a network has an anomaly zone, yet the network itself has one. 
 
 In many cases, biologists are interested in identifying the species tree despite reticulations. We demonstrated that in the presence of deep coalescence, 
 the most likely gene tree is not necessarily one of the backbone trees inside the network. Furthermore, we studied the performance of three recently 
 introduced methods in terms of their ability to recover a backbone tree inside the network. We found the none of these methods performs well when 
 deep coalescence is extensive. It is important to point out, though, that none of these methods were designed specifically for cases of hybridization, where 
 multiple genomic loci could be introgressed due to the same hybridization event. However, our findings here call for more research into the question of 
 identifying a species tree inside the network, when one exists. However, biologically, reticulation could be extensive, such as reported recently in an analysis of a 
 mosquito data set \cite{FontaineEtAl15,WenEtAl16}, in which case, designating a ``species tree" might not be adequate \cite{clark2015conundrum}. From a computational perspective, identifying such a tree aids significantly in searching for 
 networks from data \cite{yu2014maximum,WenEtAl16a} as they can serve as the starting phylogeny to which reticulation edges could be added. 
 
 Finally, many existing approaches for network inference rely on the assumption that the input gene trees are a subset of the set of trees displayed 
 by a network and, consequently, seek to infer a phylogenetic network that displays all the gene trees. In the presence of deep coalescence, this approach 
 would result in very erroneous networks. We argued that in this situation, parental trees need to be inferred first from gene trees and then a network that 
 contains the inferred parental trees could be estimated. To demonstrate the merit for this approach, we introduced a method by which gene trees are first 
 clustered and then parental trees are inferred for the clusters. The results were very promising for this clustering-based 
 approach to be pursued further. In terms of network inference, this approach gives rise to a new computational problem in which a network is sought to 
 contain a given set of parental trees. It is important to acknowledge here that our performance study of the clustering approach is very preliminary and is 
 aimed at introducing the problem and demonstrating its merit in a relatively ideal setting. We identify as a direction for future research a thorough analysis that examines, 
 among many other aspects, the 
 effects of errors in gene tree estimates (as opposed to using true gene trees), larger variations in the network's branch lengths, and the number of reticulations in 
 the network, on the performance of the approach.

%%%%%%%%%%%%%%%%%%%%%%%%%%%%%%%%%%%%%%%%%%%%%%
%%                                          %%
%% Backmatter begins here                   %%
%%                                          %%
%%%%%%%%%%%%%%%%%%%%%%%%%%%%%%%%%%%%%%%%%%%%%%

\begin{backmatter}

%\section*{Abbreviations}
%\begin{itemize}
%\item[] 
%\end{itemize}

%\section*{Competing interests}
% The authors declare that they have no competing interests.
%
%\section*{Author's contributions}
% JZ, YY and LN conceived of the study, designed the methods, analyzed the data, and wrote the manuscript. 
% JZ ran the simulations.
%
\vspace{-.2in}
\section*{Acknowledgements}
  This work was supported in part by grant CCF-1302179 from the National Science Foundation of the United 
  States of America.

%%  After submission of the .TEX file,                     %%
%%  you will be prompted to submit your .BBL file.         %%

%\newpage
% if your bibliography is in bibtex format, use those commands:
\bibliographystyle{bmc-mathphys} % Style BST file (bmc-mathphys, vancouver, spbasic).

\newcommand{\BMCxmlcomment}[1]{}

\BMCxmlcomment{

<refgrp>

<bibl id="B1">
  <title><p>Application of Phylogenetic Networks in Evolutionary
  Studies</p></title>
  <aug>
    <au><snm>Huson</snm><fnm>D.H.</fnm></au>
    <au><snm>Bryant</snm><fnm>D.</fnm></au>
  </aug>
  <source>Molecular Biology and Evolution</source>
  <pubdate>2006</pubdate>
  <volume>23</volume>
  <issue>2</issue>
  <fpage>254</fpage>
  <lpage>267</lpage>
</bibl>

<bibl id="B2">
  <title><p>Evolutionary phylogenetic networks: models and issues</p></title>
  <aug>
    <au><snm>Nakhleh</snm><fnm>L</fnm></au>
  </aug>
  <source>Problem solving handbook in computational biology and
  bioinformatics</source>
  <publisher>New York: Springer</publisher>
  <pubdate>2010</pubdate>
  <fpage>125</fpage>
  <lpage>-158</lpage>
</bibl>

<bibl id="B3">
  <title><p>Networks: expanding evolutionary thinking</p></title>
  <aug>
    <au><snm>Bapteste</snm><fnm>E</fnm></au>
    <au><snm>Iersel</snm><fnm>L</fnm></au>
    <au><snm>Janke</snm><fnm>A</fnm></au>
    <au><snm>Kelchner</snm><fnm>S</fnm></au>
    <au><snm>Kelk</snm><fnm>S</fnm></au>
    <au><snm>McInerney</snm><fnm>JO</fnm></au>
    <au><snm>Morrison</snm><fnm>DA</fnm></au>
    <au><snm>Nakhleh</snm><fnm>L</fnm></au>
    <au><snm>Steel</snm><fnm>M</fnm></au>
    <au><snm>Stougie</snm><fnm>L</fnm></au>
    <au><snm>Whitefield</snm><fnm>J</fnm></au>
  </aug>
  <source>Trends in Genetics</source>
  <publisher>Elsevier</publisher>
  <pubdate>2013</pubdate>
  <volume>29</volume>
  <issue>8</issue>
  <fpage>439</fpage>
  <lpage>441</lpage>
</bibl>

<bibl id="B4">
  <title><p>Phylogenetic Networks: Concepts, Algorithms and
  Applications</p></title>
  <aug>
    <au><snm>Huson</snm><fnm>D.H.</fnm></au>
    <au><snm>Rupp</snm><fnm>R.</fnm></au>
    <au><snm>Scornavacca</snm><fnm>C.</fnm></au>
  </aug>
  <publisher>New York: Cambridge University Press</publisher>
  <pubdate>2010</pubdate>
</bibl>

<bibl id="B5">
  <title><p>Introduction to phylogenetic networks</p></title>
  <aug>
    <au><snm>Morrison</snm><fnm>DA</fnm></au>
  </aug>
  <publisher>Sweden: RJR Productions</publisher>
  <pubdate>2011</pubdate>
</bibl>

<bibl id="B6">
  <title><p>ReCombinatorics: the algorithmics of ancestral recombination graphs
  and explicit phylogenetic networks</p></title>
  <aug>
    <au><snm>Gusfield</snm><fnm>D</fnm></au>
  </aug>
  <publisher>Boston: MIT Press</publisher>
  <pubdate>2014</pubdate>
</bibl>

<bibl id="B7">
  <title><p>Perfect phylogenetic networks with recombination</p></title>
  <aug>
    <au><snm>Wang</snm><fnm>L</fnm></au>
    <au><snm>Zhang</snm><fnm>K</fnm></au>
    <au><snm>Zhang</snm><fnm>L</fnm></au>
  </aug>
  <source>Journal of Computational Biology</source>
  <publisher>Mary Ann Liebert, Inc.</publisher>
  <pubdate>2001</pubdate>
  <volume>8</volume>
  <issue>1</issue>
  <fpage>69</fpage>
  <lpage>-78</lpage>
</bibl>

<bibl id="B8">
  <title><p>Perfect phylogenetic networks: A new methodology for reconstructing
  the evolutionary history of natural languages</p></title>
  <aug>
    <au><snm>Nakhleh</snm><fnm>L</fnm></au>
    <au><snm>Ringe</snm><fnm>D</fnm></au>
    <au><snm>Warnow</snm><fnm>T</fnm></au>
  </aug>
  <source>Language</source>
  <publisher>JSTOR</publisher>
  <pubdate>2005</pubdate>
  <fpage>382</fpage>
  <lpage>-420</lpage>
</bibl>

<bibl id="B9">
  <title><p>A decomposition theory for phylogenetic networks and incompatible
  characters</p></title>
  <aug>
    <au><snm>Gusfield</snm><fnm>D</fnm></au>
    <au><snm>Bansal</snm><fnm>V</fnm></au>
    <au><snm>Bafna</snm><fnm>V</fnm></au>
    <au><snm>Song</snm><fnm>YS</fnm></au>
  </aug>
  <source>Journal of Computational Biology</source>
  <publisher>Mary Ann Liebert, Inc. 2 Madison Avenue Larchmont, NY 10538
  USA</publisher>
  <pubdate>2007</pubdate>
  <volume>14</volume>
  <issue>10</issue>
  <fpage>1247</fpage>
  <lpage>-1272</lpage>
</bibl>

<bibl id="B10">
  <title><p>Efficient reconstruction of phylogenetic networks with constrained
  recombination</p></title>
  <aug>
    <au><snm>Gusfield</snm><fnm>D</fnm></au>
    <au><snm>Eddhu</snm><fnm>S</fnm></au>
    <au><snm>Langley</snm><fnm>C</fnm></au>
  </aug>
  <source>Bioinformatics Conference, 2003. CSB 2003. Proceedings of the 2003
  IEEE</source>
  <pubdate>2003</pubdate>
  <fpage>363</fpage>
  <lpage>-374</lpage>
</bibl>

<bibl id="B11">
  <title><p>Algorithms to distinguish the role of gene-conversion from
  single-crossover recombination in the derivation of SNP sequences in
  populations</p></title>
  <aug>
    <au><snm>Song</snm><fnm>YS</fnm></au>
    <au><snm>Ding</snm><fnm>Z</fnm></au>
    <au><snm>Gusfield</snm><fnm>D</fnm></au>
    <au><snm>Langley</snm><fnm>CH</fnm></au>
    <au><snm>Wu</snm><fnm>Y</fnm></au>
  </aug>
  <source>Research in Computational Molecular Biology</source>
  <pubdate>2006</pubdate>
  <fpage>231</fpage>
  <lpage>-245</lpage>
</bibl>

<bibl id="B12">
  <title><p>Parsimonious reconstruction of sequence evolution and haplotype
  blocks</p></title>
  <aug>
    <au><snm>Song</snm><fnm>YS</fnm></au>
    <au><snm>Hein</snm><fnm>J</fnm></au>
  </aug>
  <source>Lecture Notes in Bioinformatics</source>
  <publisher>Berlin Heidelberg: Springer</publisher>
  <pubdate>2003</pubdate>
  <volume>2812</volume>
  <fpage>287</fpage>
  <lpage>-302</lpage>
</bibl>

<bibl id="B13">
  <title><p>On the minimum number of recombination events in the evolutionary
  history of DNA sequences</p></title>
  <aug>
    <au><snm>Song</snm><fnm>YS</fnm></au>
    <au><snm>Hein</snm><fnm>J</fnm></au>
  </aug>
  <source>Journal of Mathematical Biology</source>
  <publisher>Springer</publisher>
  <pubdate>2004</pubdate>
  <volume>48</volume>
  <issue>2</issue>
  <fpage>160</fpage>
  <lpage>-186</lpage>
</bibl>

<bibl id="B14">
  <title><p>Constructing minimal ancestral recombination graphs</p></title>
  <aug>
    <au><snm>Song</snm><fnm>YS</fnm></au>
    <au><snm>Hein</snm><fnm>J</fnm></au>
  </aug>
  <source>Journal of Computational Biology</source>
  <publisher>Mary Ann Liebert, Inc. 2 Madison Avenue Larchmont, NY 10538
  USA</publisher>
  <pubdate>2005</pubdate>
  <volume>12</volume>
  <issue>2</issue>
  <fpage>147</fpage>
  <lpage>-169</lpage>
</bibl>

<bibl id="B15">
  <title><p>Reconstructing evolution of sequences subject to recombination
  using parsimony</p></title>
  <aug>
    <au><snm>Hein</snm><fnm>J.</fnm></au>
  </aug>
  <source>Mathematical Biosciences</source>
  <pubdate>1990</pubdate>
  <volume>98</volume>
  <fpage>185</fpage>
  <lpage>-200</lpage>
</bibl>

<bibl id="B16">
  <title><p>Reconstructing phylogenetic networks using maximum
  parsimony</p></title>
  <aug>
    <au><snm>Nakhleh</snm><fnm>L.</fnm></au>
    <au><snm>Jin</snm><fnm>G.</fnm></au>
    <au><snm>Zhao</snm><fnm>F.</fnm></au>
    <au><snm>Mellor Crummey</snm><fnm>J.</fnm></au>
  </aug>
  <source>Proceedings of the 2005 IEEE Computational Systems Bioinformatics
  Conference (CSB2005)</source>
  <pubdate>2005</pubdate>
  <fpage>93</fpage>
  <lpage>102</lpage>
</bibl>

<bibl id="B17">
  <title><p>Efficient parsimony-based methods for phylogenetic network
  reconstruction</p></title>
  <aug>
    <au><snm>Jin</snm><fnm>G.</fnm></au>
    <au><snm>Nakhleh</snm><fnm>L.</fnm></au>
    <au><snm>Snir</snm><fnm>S.</fnm></au>
    <au><snm>Tuller</snm><fnm>T.</fnm></au>
  </aug>
  <source>Bioinformatics</source>
  <pubdate>2006</pubdate>
  <volume>23</volume>
  <fpage>e123</fpage>
  <lpage>e128</lpage>
  <note>Proceedings of the European Conference on Computational Biology (ECCB
  06)</note>
</bibl>

<bibl id="B18">
  <title><p>Maximum likelihood of phylogenetic networks</p></title>
  <aug>
    <au><snm>Jin</snm><fnm>G.</fnm></au>
    <au><snm>Nakhleh</snm><fnm>L.</fnm></au>
    <au><snm>Snir</snm><fnm>S.</fnm></au>
    <au><snm>Tuller</snm><fnm>T.</fnm></au>
  </aug>
  <source>Bioinformatics</source>
  <pubdate>2006</pubdate>
  <volume>22</volume>
  <issue>21</issue>
  <fpage>2604</fpage>
  <lpage>2611</lpage>
</bibl>

<bibl id="B19">
  <title><p>A New Linear-time Heuristic Algorithm for Computing the Parsimony
  Score of Phylogenetic Networks: Theoretical Bounds and Empirical
  Performance</p></title>
  <aug>
    <au><snm>Jin</snm><fnm>G.</fnm></au>
    <au><snm>Nakhleh</snm><fnm>L.</fnm></au>
    <au><snm>Snir</snm><fnm>S.</fnm></au>
    <au><snm>Tuller</snm><fnm>T.</fnm></au>
  </aug>
  <source>Proceedings of the International Symposium on Bioinformatics Research
  and Applications</source>
  <editor>I. Mandoiu and A. Zelikovsky</editor>
  <pubdate>2007</pubdate>
  <inpress />
</bibl>

<bibl id="B20">
  <title><p>Inferring phylogenetic networks by the maximum parsimony criterion:
  a case study</p></title>
  <aug>
    <au><snm>Jin</snm><fnm>G.</fnm></au>
    <au><snm>Nakhleh</snm><fnm>L.</fnm></au>
    <au><snm>Snir</snm><fnm>S.</fnm></au>
    <au><snm>Tuller</snm><fnm>T.</fnm></au>
  </aug>
  <source>Mol. Biol. Evol.</source>
  <pubdate>2007</pubdate>
  <volume>24</volume>
  <issue>1</issue>
  <fpage>324</fpage>
  <lpage>337</lpage>
</bibl>

<bibl id="B21">
  <title><p>Hybrids in real time</p></title>
  <aug>
    <au><snm>Baroni</snm><fnm>M.</fnm></au>
    <au><snm>Semple</snm><fnm>C.</fnm></au>
    <au><snm>Steel</snm><fnm>M.</fnm></au>
  </aug>
  <source>Syst. Biol.</source>
  <pubdate>2006</pubdate>
  <volume>55</volume>
  <issue>1</issue>
  <fpage>46</fpage>
  <lpage>56</lpage>
</bibl>

<bibl id="B22">
  <title><p>Summarizing Multiple Gene Trees Using Cluster Networks</p></title>
  <aug>
    <au><snm>Huson</snm><fnm>D.H.</fnm></au>
    <au><snm>Rupp</snm><fnm>R.</fnm></au>
  </aug>
  <source>Proceedings of the Workshop on Algorithms in Bioinformatics</source>
  <editor>K.A. Crandall and J. Lagergren</editor>
  <series><title><p>Lecture Notes in Bioinformatics</p></title></series>
  <pubdate>2008</pubdate>
  <volume>5251</volume>
  <fpage>296</fpage>
  <lpage>305</lpage>
</bibl>

<bibl id="B23">
  <title><p>Phylogenetic networks do not need to be complex: using fewer
  reticulations to represent conflicting clusters</p></title>
  <aug>
    <au><snm>Van Iersel</snm><fnm>L</fnm></au>
    <au><snm>Kelk</snm><fnm>S</fnm></au>
    <au><snm>Rupp</snm><fnm>R</fnm></au>
    <au><snm>Huson</snm><fnm>D</fnm></au>
  </aug>
  <source>Bioinformatics</source>
  <publisher>Oxford Univ Press</publisher>
  <pubdate>2010</pubdate>
  <volume>26</volume>
  <issue>12</issue>
  <fpage>i124</fpage>
  <lpage>-i131</lpage>
</bibl>

<bibl id="B24">
  <title><p>An algorithm for constructing parsimonious hybridization networks
  with multiple phylogenetic trees</p></title>
  <aug>
    <au><snm>Wu</snm><fnm>Y</fnm></au>
  </aug>
  <source>Journal of Computational Biology</source>
  <publisher>Mary Ann Liebert, Inc. 140 Huguenot Street, 3rd Floor New
  Rochelle, NY 10801 USA</publisher>
  <pubdate>2013</pubdate>
  <volume>20</volume>
  <issue>10</issue>
  <fpage>792</fpage>
  <lpage>-804</lpage>
</bibl>

<bibl id="B25">
  <title><p>Reconstructible phylogenetic networks: do not distinguish the
  indistinguishable</p></title>
  <aug>
    <au><snm>Pardi</snm><fnm>F</fnm></au>
    <au><snm>Scornavacca</snm><fnm>C</fnm></au>
  </aug>
  <source>PLoS Comput Biol</source>
  <publisher>Public Library of Science</publisher>
  <pubdate>2015</pubdate>
  <volume>11</volume>
  <issue>4</issue>
  <fpage>e1004135</fpage>
</bibl>

<bibl id="B26">
  <title><p>Reconstructing evolution of natural languages: Complexity and
  parameterized algorithms</p></title>
  <aug>
    <au><snm>Kanj</snm><fnm>IA</fnm></au>
    <au><snm>Nakhleh</snm><fnm>L</fnm></au>
    <au><snm>Xia</snm><fnm>G</fnm></au>
  </aug>
  <source>Computing and Combinatorics</source>
  <publisher>New York: Springer</publisher>
  <pubdate>2006</pubdate>
  <fpage>299</fpage>
  <lpage>-308</lpage>
</bibl>

<bibl id="B27">
  <title><p>Computing the hybridization number of two phylogenetic trees is
  fixed-parameter tractable</p></title>
  <aug>
    <au><snm>Bordewich</snm><fnm>M.</fnm></au>
    <au><snm>Semple</snm><fnm>C.</fnm></au>
  </aug>
  <source>IEEE/ACM Transactions on Computational Biology and
  Bioinformatics</source>
  <pubdate>2007</pubdate>
  <inpress />
</bibl>

<bibl id="B28">
  <title><p>Seeing the trees and their branches in the network is
  hard</p></title>
  <aug>
    <au><snm>Kanj</snm><fnm>IA</fnm></au>
    <au><snm>Nakhleh</snm><fnm>L</fnm></au>
    <au><snm>Than</snm><fnm>C</fnm></au>
    <au><snm>Xia</snm><fnm>G</fnm></au>
  </aug>
  <source>Theoretical Computer Science</source>
  <publisher>Elsevier</publisher>
  <pubdate>2008</pubdate>
  <volume>401</volume>
  <issue>1</issue>
  <fpage>153</fpage>
  <lpage>-164</lpage>
</bibl>

<bibl id="B29">
  <title><p>The compatibility of binary characters on phylogenetic networks:
  complexity and parameterized algorithms</p></title>
  <aug>
    <au><snm>Kanj</snm><fnm>IA</fnm></au>
    <au><snm>Nakhleh</snm><fnm>L</fnm></au>
    <au><snm>Xia</snm><fnm>G</fnm></au>
  </aug>
  <source>Algorithmica</source>
  <publisher>Springer</publisher>
  <pubdate>2008</pubdate>
  <volume>51</volume>
  <issue>2</issue>
  <fpage>99</fpage>
  <lpage>-128</lpage>
</bibl>

<bibl id="B30">
  <title><p>Locating a tree in a phylogenetic network</p></title>
  <aug>
    <au><snm>Van Iersel</snm><fnm>L</fnm></au>
    <au><snm>Semple</snm><fnm>C</fnm></au>
    <au><snm>Steel</snm><fnm>M</fnm></au>
  </aug>
  <source>Information Processing Letters</source>
  <publisher>Elsevier</publisher>
  <pubdate>2010</pubdate>
  <volume>110</volume>
  <issue>23</issue>
  <fpage>1037</fpage>
  <lpage>-1043</lpage>
</bibl>

<bibl id="B31">
  <title><p>When two trees go to war</p></title>
  <aug>
    <au><snm>Van Iersel</snm><fnm>L</fnm></au>
    <au><snm>Kelk</snm><fnm>S</fnm></au>
  </aug>
  <source>Journal of theoretical biology</source>
  <publisher>Elsevier</publisher>
  <pubdate>2011</pubdate>
  <volume>269</volume>
  <issue>1</issue>
  <fpage>245</fpage>
  <lpage>-255</lpage>
</bibl>

<bibl id="B32">
  <title><p>Identifying a species tree subject to random lateral gene
  transfer</p></title>
  <aug>
    <au><snm>Steel</snm><fnm>M</fnm></au>
    <au><snm>Linz</snm><fnm>S</fnm></au>
    <au><snm>Huson</snm><fnm>DH</fnm></au>
    <au><snm>Sanderson</snm><fnm>MJ</fnm></au>
  </aug>
  <source>Journal of theoretical biology</source>
  <publisher>Elsevier</publisher>
  <pubdate>2013</pubdate>
  <volume>322</volume>
  <fpage>81</fpage>
  <lpage>-93</lpage>
</bibl>

<bibl id="B33">
  <title><p>Species Trees from Gene Trees Despite a High Rate of Lateral
  Genetic Transfer: A Tight Bound</p></title>
  <aug>
    <au><snm>Daskalakis</snm><fnm>C</fnm></au>
    <au><snm>Roch</snm><fnm>S</fnm></au>
  </aug>
  <source>arXiv preprint arXiv:1508.01962</source>
  <pubdate>2015</pubdate>
</bibl>

<bibl id="B34">
  <title><p>Phylogenomic species tree estimation in the presence of incomplete
  lineage sorting and horizontal gene transfer</p></title>
  <aug>
    <au><snm>Davidson</snm><fnm>R</fnm></au>
    <au><snm>Vachaspati</snm><fnm>P</fnm></au>
    <au><snm>Mirarab</snm><fnm>S</fnm></au>
    <au><snm>Warnow</snm><fnm>T</fnm></au>
  </aug>
  <source>BMC genomics</source>
  <publisher>BioMed Central Ltd</publisher>
  <pubdate>2015</pubdate>
  <volume>16</volume>
  <issue>Suppl 10</issue>
  <fpage>S1</fpage>
</bibl>

<bibl id="B35">
  <title><p>Which phylogenetic networks are merely trees with additional
  arcs?</p></title>
  <aug>
    <au><snm>Francis</snm><fnm>AR</fnm></au>
    <au><snm>Steel</snm><fnm>M</fnm></au>
  </aug>
  <source>Systematic biology</source>
  <publisher>Oxford University Press</publisher>
  <pubdate>2015</pubdate>
  <volume>64</volume>
  <issue>5</issue>
  <fpage>768</fpage>
  <lpage>-777</lpage>
</bibl>

<bibl id="B36">
  <title><p>The probability of topological concordance of gene trees and
  species trees</p></title>
  <aug>
    <au><snm>Rosenberg</snm><fnm>NA</fnm></au>
  </aug>
  <source>Theoretical population biology</source>
  <publisher>Elsevier</publisher>
  <pubdate>2002</pubdate>
  <volume>61</volume>
  <issue>2</issue>
  <fpage>225</fpage>
  <lpage>-247</lpage>
</bibl>

<bibl id="B37">
  <title><p>Inconsistency of species-tree methods under gene flow</p></title>
  <aug>
    <au><snm>Sol{\'\i}s Lemus</snm><fnm>C</fnm></au>
    <au><snm>Yang</snm><fnm>M</fnm></au>
    <au><snm>An{\'e}</snm><fnm>C</fnm></au>
  </aug>
  <source>Systematic biology</source>
  <publisher>Oxford University Press</publisher>
  <pubdate>2016</pubdate>
  <fpage>syw030</fpage>
</bibl>

<bibl id="B38">
  <title><p>Clustering Genes of Common Evolutionary History</p></title>
  <aug>
    <au><snm>Gori</snm><fnm>K</fnm></au>
    <au><snm>Suchan</snm><fnm>T</fnm></au>
    <au><snm>Alvarez</snm><fnm>N</fnm></au>
    <au><snm>Goldman</snm><fnm>N</fnm></au>
    <au><snm>Dessimoz</snm><fnm>C</fnm></au>
  </aug>
  <pubdate>2016</pubdate>
  <volume>33</volume>
  <issue>6</issue>
  <fpage>1590</fpage>
  <lpage>1605</lpage>
</bibl>

<bibl id="B39">
  <title><p>Gene tree distributions under the coalescent process</p></title>
  <aug>
    <au><snm>Degnan</snm><fnm>J.H.</fnm></au>
    <au><snm>Salter</snm><fnm>L.A.</fnm></au>
  </aug>
  <source>Evolution</source>
  <pubdate>2005</pubdate>
  <volume>59</volume>
  <fpage>24</fpage>
  <lpage>37</lpage>
</bibl>

<bibl id="B40">
  <title><p>The probability of a gene tree topology within a phylogenetic
  network with applications to hybridization detection</p></title>
  <aug>
    <au><snm>Yu</snm><fnm>Y</fnm></au>
    <au><snm>Degnan</snm><fnm>JH</fnm></au>
    <au><snm>Nakhleh</snm><fnm>L</fnm></au>
  </aug>
  <source>PLoS Genet</source>
  <publisher>Public Library of Science</publisher>
  <pubdate>2012</pubdate>
  <volume>8</volume>
  <issue>4</issue>
  <fpage>e1002660</fpage>
</bibl>

<bibl id="B41">
  <title><p>Maximum likelihood inference of reticulate evolutionary
  histories</p></title>
  <aug>
    <au><snm>Yu</snm><fnm>Y</fnm></au>
    <au><snm>Dong</snm><fnm>J</fnm></au>
    <au><snm>Liu</snm><fnm>KJ</fnm></au>
    <au><snm>Nakhleh</snm><fnm>L</fnm></au>
  </aug>
  <source>Proceedings of the National Academy of Sciences</source>
  <publisher>National Acad Sciences</publisher>
  <pubdate>2014</pubdate>
  <volume>111</volume>
  <issue>46</issue>
  <fpage>16448</fpage>
  <lpage>-16453</lpage>
</bibl>

<bibl id="B42">
  <title><p>Discordance of species trees with their most likely gene
  trees</p></title>
  <aug>
    <au><snm>Degnan</snm><fnm>JH</fnm></au>
    <au><snm>Rosenberg</snm><fnm>NA</fnm></au>
  </aug>
  <source>PLoS Genet</source>
  <publisher>Public Library of Science</publisher>
  <pubdate>2006</pubdate>
  <volume>2</volume>
  <issue>5</issue>
  <fpage>e68</fpage>
</bibl>

<bibl id="B43">
  <title><p>Extensive introgression in a malaria vector species complex
  revealed by phylogenomics</p></title>
  <aug>
    <au><snm>Fontaine</snm><fnm>MC</fnm></au>
    <au><snm>Pease</snm><fnm>JB</fnm></au>
    <au><snm>Steele</snm><fnm>A</fnm></au>
    <au><snm>Waterhouse</snm><fnm>RM</fnm></au>
    <au><snm>Neafsey</snm><fnm>DE</fnm></au>
    <au><snm>Sharakhov</snm><fnm>IV</fnm></au>
    <au><snm>Jiang</snm><fnm>X</fnm></au>
    <au><snm>Hall</snm><fnm>AB</fnm></au>
    <au><snm>Catteruccia</snm><fnm>F</fnm></au>
    <au><snm>Kakani</snm><fnm>E</fnm></au>
    <au><cnm>others</cnm></au>
  </aug>
  <source>Science</source>
  <publisher>American Association for the Advancement of Science</publisher>
  <pubdate>2015</pubdate>
  <volume>347</volume>
  <issue>6217</issue>
  <fpage>1258524</fpage>
</bibl>

<bibl id="B44">
  <title><p>{ASTRAL-II}: coalescent-based species tree estimation with many
  hundreds of taxa and thousands of genes</p></title>
  <aug>
    <au><snm>Mirarab</snm><fnm>S</fnm></au>
    <au><snm>Warnow</snm><fnm>T</fnm></au>
  </aug>
  <source>Bioinformatics</source>
  <publisher>Oxford Univ Press</publisher>
  <pubdate>2015</pubdate>
  <volume>31</volume>
  <issue>12</issue>
  <fpage>i44</fpage>
  <lpage>-i52</lpage>
</bibl>

<bibl id="B45">
  <title><p>Dendroscope 3: an interactive tool for rooted phylogenetic trees
  and networks</p></title>
  <aug>
    <au><snm>Huson</snm><fnm>DH</fnm></au>
    <au><snm>Scornavacca</snm><fnm>C</fnm></au>
  </aug>
  <source>Systematic biology</source>
  <publisher>Oxford University Press</publisher>
  <pubdate>2012</pubdate>
  <volume>61</volume>
  <issue>6</issue>
  <fpage>1061</fpage>
  <lpage>-1067</lpage>
</bibl>

<bibl id="B46">
  <title><p>Generating samples under a Wright--Fisher neutral model of genetic
  variation</p></title>
  <aug>
    <au><snm>Hudson</snm><fnm>RR</fnm></au>
  </aug>
  <source>Bioinformatics</source>
  <publisher>Oxford Univ Press</publisher>
  <pubdate>2002</pubdate>
  <volume>18</volume>
  <issue>2</issue>
  <fpage>337</fpage>
  <lpage>-338</lpage>
</bibl>

<bibl id="B47">
  <title><p>Comparison of phylogenetic trees</p></title>
  <aug>
    <au><snm>Robinson</snm><fnm>D.R.</fnm></au>
    <au><snm>Foulds</snm><fnm>L.R.</fnm></au>
  </aug>
  <source>Math. Biosci.</source>
  <pubdate>1981</pubdate>
  <volume>53</volume>
  <fpage>131</fpage>
  <lpage>-147</lpage>
</bibl>

<bibl id="B48">
  <title><p>MDSJ: Java Library for Multidimensional Scaling (Version
  0.2).</p></title>
  <aug>
    <au><snm>Group.</snm><fnm>A</fnm></au>
  </aug>
  <source>Available at
  \url{http://www.inf.uni-konstanz.de/algo/software/mdsj/}.</source>
  <note>Algorithmics Group, University of Konstanz, 2009.</note>
</bibl>

<bibl id="B49">
  <title><p>Some methods for classification and analysis of multivariate
  observations</p></title>
  <aug>
    <au><snm>MacQueen</snm><fnm>J</fnm></au>
    <au><cnm>others</cnm></au>
  </aug>
  <source>Proceedings of the fifth Berkeley symposium on mathematical
  statistics and probability</source>
  <pubdate>1967</pubdate>
  <volume>1</volume>
  <issue>14</issue>
  <fpage>281</fpage>
  <lpage>-297</lpage>
</bibl>

<bibl id="B50">
  <title><p>Silhouettes: a graphical aid to the interpretation and validation
  of cluster analysis</p></title>
  <aug>
    <au><snm>Rousseeuw</snm><fnm>PJ</fnm></au>
  </aug>
  <source>Journal of computational and applied mathematics</source>
  <publisher>Elsevier</publisher>
  <pubdate>1987</pubdate>
  <volume>20</volume>
  <fpage>53</fpage>
  <lpage>-65</lpage>
</bibl>

<bibl id="B51">
  <title><p>Species tree inference by minimizing deep coalescences</p></title>
  <aug>
    <au><snm>Than</snm><fnm>C.</fnm></au>
    <au><snm>Nakhleh</snm><fnm>L.</fnm></au>
  </aug>
  <source>PLoS Computational Biology</source>
  <pubdate>2009</pubdate>
  <volume>5</volume>
  <issue>9</issue>
  <fpage>e1000501</fpage>
</bibl>

<bibl id="B52">
  <title><p>{PhyloNet}: a software package for analyzing and reconstructing
  reticulate evolutionary relationships</p></title>
  <aug>
    <au><snm>Than</snm><fnm>C.</fnm></au>
    <au><snm>Ruths</snm><fnm>D.</fnm></au>
    <au><snm>Nakhleh</snm><fnm>L.</fnm></au>
  </aug>
  <source>BMC Bioinformatics</source>
  <pubdate>2008</pubdate>
  <volume>9</volume>
  <fpage>322</fpage>
</bibl>

<bibl id="B53">
  <title><p>Towards the development of computational tools for evaluating
  phylogenetic network reconstruction methods</p></title>
  <aug>
    <au><snm>Warnow</snm><fnm>LNJST</fnm></au>
    <au><snm>Linder</snm><fnm>CR</fnm></au>
    <au><snm>Tholse</snm><fnm>BMMA</fnm></au>
  </aug>
  <source>Proc. Eighth Pacific Symp. Biocomputing (PSB'03)</source>
  <publisher>Singapore: World Scientific Publishing</publisher>
  <pubdate>2003</pubdate>
  <fpage>315</fpage>
  <lpage>-326</lpage>
</bibl>

<bibl id="B54">
  <title><p>Reticulate evolutionary history and extensive introgression in
  mosquito species revealed by phylogenetic network analysis</p></title>
  <aug>
    <au><snm>Wen</snm><fnm>D.</fnm></au>
    <au><snm>Yu</snm><fnm>Y.</fnm></au>
    <au><snm>Hahn</snm><fnm>M.W.</fnm></au>
    <au><snm>Nakhleh</snm><fnm>L.</fnm></au>
  </aug>
  <source>Molecular Ecology</source>
  <pubdate>2016</pubdate>
  <volume>25</volume>
  <fpage>2361</fpage>
  <lpage>2372</lpage>
</bibl>

<bibl id="B55">
  <title><p>Conundrum of jumbled mosquito genomes</p></title>
  <aug>
    <au><snm>Clark</snm><fnm>AG</fnm></au>
    <au><snm>Messer</snm><fnm>PW</fnm></au>
    <au><cnm>others</cnm></au>
  </aug>
  <source>Science</source>
  <publisher>American Association for The Advancement of Science, Washington,
  USA</publisher>
  <pubdate>2015</pubdate>
  <volume>347</volume>
  <issue>6217</issue>
  <fpage>27</fpage>
  <lpage>-28</lpage>
</bibl>

<bibl id="B56">
  <title><p>Bayesian Inference of Reticulate Phylogenies under the Multispecies
  Network Coalescent</p></title>
  <aug>
    <au><snm>Wen</snm><fnm>D.</fnm></au>
    <au><snm>Yu</snm><fnm>Y.</fnm></au>
    <au><snm>Nakhleh</snm><fnm>L.</fnm></au>
  </aug>
  <source>PLoS Genetics</source>
  <pubdate>2016</pubdate>
  <volume>12</volume>
  <issue>5</issue>
  <fpage>e1006006</fpage>
</bibl>

</refgrp>
} % end of \BMCxmlcomment

\end{backmatter}

\newpage
\section*{APPENDIX}

 The pseudo-code in Algorithm \ref{alg:Net2Tre} below is that of the algorithm for computing the MUL-tree 
 $T$ of a phylogenetic network $\psi$. 

 \begin{algorithm}[H]
 {\small
 \KwIn{Phylogenetic $\X$-network $\psi$ and its branch lengths {\boldmath $\lambda$}.}
 \KwOut{MUL-tree $T$ and its branch lengths {\boldmath $\lambda'$}.}
 $T \leftarrow \psi$\;  
 {\boldmath $\lambda'$}$\leftarrow${\boldmath $\lambda$}\;
  \While{traversing the nodes of $T$ bottom-up}{ 
 \If{node $h$ has two parents, $u$ and $v$, and child $w$}{
   Create a copy of $T_w$ whose root is new node $w'$\;
   Add to $T$ two new edges $e_1=(u,w)$ and $e_2=(v,w')$\;
   ${\lambda'}_{(u,w)} \leftarrow {\lambda}_{(u,h)}+{\lambda}_{(h,w)}$; ${\lambda'}_{(v,w)} \leftarrow  {\lambda}_{(v,h)}+{\lambda}_{(h,w)}$\;
   Delete from $T$ node $h$ and edges $(u,h)$, $(v,h)$, and $(h,w)$\;
   Delete $\lambda'_{(u,h)}$, $\lambda'_{(v,h)}$, $\lambda'_{(h,w)}$\; 
 }
 }
 \Return{$T$}\;
 \caption{{\bf NetworkToMULTree.} \label{alg:Net2Tre}}
 }
 \end{algorithm}

\begin{proof}[Proof of Lemma \ref{lemm1}]
 Let $\psi$ be a phylogenetic networks on $3$ taxa, and consider the set ${\cal W}(\psi)$ when restricted only to the distinct topologies. 
 We have $1 \leq |{\cal W}(\psi)| \leq 3$.  

If $|{\cal W}(\psi)|=3$, then 
the topology of every gene tree on the same set of $3$ taxa is an element of ${\cal W}(\psi)$. Therefore, no gene 
tree can satisfy Eq. \eqref{anom}. 

If $|{\cal W}(\psi)|=2$, without loss of generality, let the two parental trees be $((A,B),C)$ and $(A,(B,C))$. If $\psi$ produces an 
anomaly, then it must be that the anomalous gene tree is $((a,c),b)$. To obtain this gene tree, $a$ and $c$ must coalesce 
above the root in both parental trees. Since for the other two gene trees the coalescence events could occur under or above 
the root, the probability of each of them is at least the probability of $((a,c),b)$. Therefore, $((a,c),b)$ is not anomalous. 

If $|{\cal W}(\psi)| = 1$, without loss of generality, let the parental tree topology be $((A,B),C)$. If $\psi$ produces an anomaly, 
then it must be that the anomalous gene tree is either $((a,c),b)$ or $(a,(b,c))$. To obtain $((a,c),b)$, $a$ and $c$ must coalesce 
above the root in the parental tree. And to obtain $(a,(b,c))$, $b$ and $c$ must also coalesce above the root in the parental tree. Since for $((a,b),c)$ the coalescence events could occur under or above the root, its probability is at least the probability of $((a,c),b)$ (or $(a,(b,c))$ ). Therefore neither $((a,c),b)$ nor $(a,(b,c))$ is anomalous.
\end{proof}

% or include bibliography directly:
% \begin{thebibliography}
% \bibitem{b1}
% \end{thebibliography}

%%%%%%%%%%%%%%%%%%%%%%%%%%%%%%%%%%%
%%                               %%
%% Figures                       %%
%%                               %%
%% NB: this is for captions and  %%
%% Titles. All graphics must be  %%
%% submitted separately and NOT  %%
%% included in the Tex document  %%
%%                               %%
%%%%%%%%%%%%%%%%%%%%%%%%%%%%%%%%%%%

%%
%% Do not use \listoffigures as most will included as separate files
%
%\section*{Figures}
%  \begin{figure}[h!]
%  \caption{\csentence{Sample figure title.}
%      A short description of the figure content
%      should go here.}
%      \end{figure}
%
%\begin{figure}[h!]
%  \caption{\csentence{Sample figure title.}
%      Figure legend text.}
%      \end{figure}

%%%%%%%%%%%%%%%%%%%%%%%%%%%%%%%%%%%
%%                               %%
%% Tables                        %%
%%                               %%
%%%%%%%%%%%%%%%%%%%%%%%%%%%%%%%%%%%

% Use of \listoftables is discouraged.
%

%%%%%%%%%%%%%%%%%%%%%%%%%%%%%%%%%%%
%%                               %%
%% Additional Files              %%
%%                               %%
%%%%%%%%%%%%%%%%%%%%%%%%%%%%%%%%%%%
%
%\section*{Additional Files}
%  \subsection*{Additional file 1 --- Sample additional file title}
%    Additional file descriptions text (including details of how to
%    view the file, if it is in a non-standard format or the file extension).  This might
%    refer to a multi-page table or a figure.
%
%  \subsection*{Additional file 2 --- Sample additional file title}
%    Additional file descriptions text.

\begin{table}[!ht]
\centering
\caption{Probabilities of 15 rooted gene trees given the phylogenetic network $\psi$ of Fig. \ref{fig:dc}(b) ($w=0$). The quantity $g_{ij}(t)$ is the probability that $i$ 
lineages coalesce into $j$ lineages within time $t$ \cite{rosenberg2002probability}.}
\label{tab:table1}
\begin{tabular}{|c|c|} 
\hline
Gene Tree $T_i$       & $P(T_i|\psi,x,y,\gamma)$ \\ \hline\hline
$T_1=(((b,c),a),d)$                & \tabincell{c}{$g_{21}(y)[\gamma(g_{21}(x)+g_{22}(x)\frac{1}{3})+(1-\gamma)(g_{22}(x)\frac{1}{3})]$ \\ $+g_{22}(y)[\gamma^2(g_{31}(x)\frac{1}{3}+g_{32}(x)\frac{1}{3}\frac{1}{3}+g_{33}(x)\frac{1}{6}\frac{1}{3})$ \\ $+(1-\gamma)^2(g_{32}(x)\frac{1}{3}\frac{1}{3}+g_{33}(x)\frac{1}{6}\frac{1}{3})$ \\ $+2\gamma(1-\gamma)(g_{22}(x)g_{22}(x)\frac{1}{6}\frac{1}{3})]$}     \\ \hline

$T_2=(((b,c),d),a)$                & \tabincell{c}{$g_{21}(y)[(1-\gamma)(g_{21}(x)+g_{22}(x)\frac{1}{3})+\gamma(g_{22}(x)\frac{1}{3})]$ \\ $+g_{22}(y)[(1-\gamma)^2(g_{31}(x)\frac{1}{3}+g_{32}(x)\frac{1}{3}\frac{1}{3}+g_{33}(x)\frac{1}{6}\frac{1}{3})$ \\ $+\gamma^2(g_{32}(x)\frac{1}{3}\frac{1}{3}+g_{33}(x)\frac{1}{6}\frac{1}{3})$ \\ $+2\gamma(1-\gamma)(g_{22}(x)g_{22}(x)\frac{1}{6}\frac{1}{3})]$}     \\ \hline

$T_3=((a,b),(c,d))$                & \tabincell{c}{$g_{22}(y)[(\gamma^2+(1-\gamma)^2)(g_{32}(x)\frac{1}{3}\frac{1}{3}+g_{33}(x)\frac{2}{6}\frac{1}{3})$ \\ $+\gamma(1-\gamma)(g_{21}(x)g_{21}(x)+g_{21}(x)g_{22}(x)\frac{1}{3}+g_{22}(x)g_{21}(x)\frac{1}{3}+g_{22}(x)g_{22}(x)\frac{2}{6}\frac{1}{3})$ \\ $+\gamma(1-\gamma)(g_{22}(x)g_{22}(x)\frac{2}{6}\frac{1}{3})]$ }      \\ \hline

$T_4=((a,c),(b,d))$                & \tabincell{c}{$g_{22}(y)[(\gamma^2+(1-\gamma)^2)(g_{32}(x)\frac{1}{3}\frac{1}{3}+g_{33}(x)\frac{2}{6}\frac{1}{3})$ \\ $+\gamma(1-\gamma)(g_{21}(x)g_{21}(x)+g_{21}(x)g_{22}(x)\frac{1}{3}+g_{22}(x)g_{21}(x)\frac{1}{3}+g_{22}(x)g_{22}(x)\frac{2}{6}\frac{1}{3})$ \\ $+\gamma(1-\gamma)(g_{22}(x)g_{22}(x)\frac{2}{6}\frac{1}{3})]$ }     \\ \hline

$T_5=(((a,b),c),d)$                & \tabincell{c}{$g_{22}(y)[\gamma^2(g_{31}(x)\frac{1}{3}+g_{32}(x)\frac{1}{3}\frac{1}{3}+g_{33}(x)\frac{1}{6}\frac{1}{3})+(1-\gamma)^2(g_{33}(x)\frac{1}{6}\frac{1}{3})$ \\ $+\gamma(1-\gamma)(g_{21}(x)g_{22}(x)\frac{1}{3}+g_{22}(x)g_{22}(x)\frac{1}{6}\frac{1}{3})+\gamma(1-\gamma)g_{22}(x)g_{22}(x)\frac{1}{6}\frac{1}{3}]$}     \\ \hline

$T_6=(((a,c),b),d)$                & \tabincell{c}{$g_{22}(y)[\gamma^2(g_{31}(x)\frac{1}{3}+g_{32}(x)\frac{1}{3}\frac{1}{3}+g_{33}(x)\frac{1}{6}\frac{1}{3})+(1-\gamma)^2(g_{33}(x)\frac{1}{6}\frac{1}{3})$ \\ $+\gamma(1-\gamma)(g_{21}(x)g_{22}(x)\frac{1}{3}+g_{22}(x)g_{22}(x)\frac{1}{6}\frac{1}{3})+\gamma(1-\gamma)g_{22}(x)g_{22}(x)\frac{1}{6}\frac{1}{3}]$}     \\ \hline

$T_7=(a,(b,(c,d)))$                & \tabincell{c}{$g_{22}(y)[(1-\gamma)^2(g_{31}(x)\frac{1}{3}+g_{32}(x)\frac{1}{3}\frac{1}{3}+g_{33}(x)\frac{1}{6}\frac{1}{3})+\gamma^2(g_{33}(x)\frac{1}{6}\frac{1}{3})$ \\ $+\gamma(1-\gamma)(g_{21}(x)g_{22}(x)\frac{1}{3}+g_{22}(x)g_{22}(x)\frac{1}{6}\frac{1}{3})+\gamma(1-\gamma)g_{22}(x)g_{22}(x)\frac{1}{6}\frac{1}{3}]$}     \\ \hline

$T_8=(((b,d),c),a)$                & \tabincell{c}{$g_{22}(y)[(1-\gamma)^2(g_{31}(x)\frac{1}{3}+g_{32}(x)\frac{1}{3}\frac{1}{3}+g_{33}(x)\frac{1}{6}\frac{1}{3})+\gamma^2(g_{33}(x)\frac{1}{6}\frac{1}{3})$ \\ $+\gamma(1-\gamma)(g_{21}(x)g_{22}(x)\frac{1}{3}+g_{22}(x)g_{22}(x)\frac{1}{6}\frac{1}{3})+\gamma(1-\gamma)g_{22}(x)g_{22}(x)\frac{1}{6}\frac{1}{3}]$}     \\ \hline

$T_9=((a,d),(b,c))$                & \tabincell{c}{$g_{21}(y)[\gamma g_{22}(x)\frac{1}{3}+(1-\gamma)g_{22}(x)\frac{1}{3}]$ \\ $+g_{22}(y)[\gamma^2(g_{32}(x)\frac{1}{3}\frac{1}{3}+g_{33}(x)\frac{2}{6}\frac{1}{3})+(1-\gamma)^2(g_{32}(x)\frac{1}{3}\frac{1}{3}+g_{33}(x)\frac{2}{6}\frac{1}{3})$ \\ $+\gamma(1-\gamma)(g_{22}(x)g_{22}(x)\frac{2}{6}\frac{1}{3})+\gamma(1-\gamma)(g_{22}(x)g_{22}(x)\frac{2}{6}\frac{1}{3})]$}    \\ \hline

$T_{10}=(((a,b),d),c)$             & \tabincell{c}{$g_{22}(y)[\gamma^2(g_{32}(x)\frac{1}{3}\frac{1}{3}+g_{33}(x)\frac{1}{6}\frac{1}{3})+(1-\gamma)^2(g_{33}(x)\frac{1}{6}\frac{1}{3})$ \\ $+\gamma(1-\gamma)(g_{21}(x)g_{22}(x)\frac{1}{3}+g_{22}(x)g_{22}(x)\frac{1}{6}\frac{1}{3})+\gamma(1-\gamma)(g_{22}(x)g_{22}(x)\frac{1}{6}\frac{1}{3})]$}     \\ \hline

$T_{11}=(b,(a,(c,d)))$             & \tabincell{c}{$g_{22}(y)[(1-\gamma)^2(g_{32}(x)\frac{1}{3}\frac{1}{3}+g_{33}(x)\frac{1}{6}\frac{1}{3})+\gamma^2(g_{33}(x)\frac{1}{6}\frac{1}{3})$ \\ $+\gamma(1-\gamma)(g_{21}(x)g_{22}(x)\frac{1}{3}+g_{22}(x)g_{22}(x)\frac{1}{6}\frac{1}{3})+\gamma(1-\gamma)(g_{22}(x)g_{22}(x)\frac{1}{6}\frac{1}{3})]$}     \\ \hline

$T_{12}=(((a,d),b),c)$             & \tabincell{c}{$g_{22}(y)[\gamma^2(g_{33}(x)\frac{1}{6}\frac{1}{3})+(1-\gamma)^2(g_{33}(x)\frac{1}{6}\frac{1}{3})$ \\ $+\gamma(1-\gamma)(g_{22}(x)g_{22}(x)\frac{1}{6}\frac{1}{3})+\gamma(1-\gamma)(g_{22}(x)g_{22}(x)\frac{1}{6}\frac{1}{3})]$}     \\ \hline

$T_{13}=(((b,d),a),c)$             & \tabincell{c}{$g_{22}(y)[\gamma^2(g_{33}(x)\frac{1}{6}\frac{1}{3})+(1-\gamma)^2(g_{32}(x)\frac{1}{3}\frac{1}{3}+g_{33}(x)\frac{1}{6}\frac{1}{3})$ \\ $+\gamma(1-\gamma)(g_{22}(x)g_{22}(x)\frac{1}{6}\frac{1}{3})+\gamma(1-\gamma)(g_{21}(x)g_{22}(x)\frac{1}{3}+g_{22}(x)g_{22}(x)\frac{1}{6}\frac{1}{3})]$}     \\ \hline

$T_{14}=(((a,c),d),b)$             & \tabincell{c}{$g_{22}(y)[(1-\gamma)^2(g_{33}(x)\frac{1}{6}\frac{1}{3})+\gamma^2(g_{32}(x)\frac{1}{3}\frac{1}{3}+g_{33}(x)\frac{1}{6}\frac{1}{3})$ \\ $+\gamma(1-\gamma)(g_{22}(x)g_{22}(x)\frac{1}{6}\frac{1}{3})+\gamma(1-\gamma)(g_{21}(x)g_{22}(x)\frac{1}{3}+g_{22}(x)g_{22}(x)\frac{1}{6}\frac{1}{3})]$}     \\ \hline

$T_{15}=(((a,d),c),b)$             & \tabincell{c}{$g_{22}(y)[\gamma^2(g_{33}(x)\frac{1}{6}\frac{1}{3})+(1-\gamma)^2(g_{33}(x)\frac{1}{6}\frac{1}{3})$ \\ $+\gamma(1-\gamma)(g_{22}(x)g_{22}(x)\frac{1}{6}\frac{1}{3})+\gamma(1-\gamma)(g_{22}(x)g_{22}(x)\frac{1}{6}\frac{1}{3})]$}     \\ \hline

\end{tabular}
\end{table}

\begin{table}[!ht]
\centering
\caption{Probabilities of 15 rooted gene trees given the phylogenetic network $\psi$ of Fig. \ref{fig:dc}(b) ($w=0$) as $x\to\infty$.}
\label{tab:table1x}
\begin{tabular}{|c|c|} 
\hline
Gene Tree $T_i$       & $P(T_i|\psi,y,\gamma)$ \\ \hline\hline
$T_1=(((b,c),a),d)$                & $\gamma-(\gamma-\frac{\gamma^2}{3})e^{-y}$     \\ \hline
$T_2=(((b,c),d),a)$                & $(1-\gamma)-(-\frac{\gamma^2}{3}-\frac{\gamma}{3}+\frac{2}{3})e^{-y}$     \\ \hline
$T_3=((a,b),(c,d))$                & $\gamma(1-\gamma)e^{-y}$                  \\ \hline
$T_4=((a,c),(b,d))$                & $\gamma(1-\gamma)e^{-y}$                  \\ \hline
$T_5=(((a,b),c),d)$                & $\frac{\gamma^2}{3}e^{-y}$     \\ \hline
$T_6=(((a,c),b),d)$                & $\frac{\gamma^2}{3}e^{-y}$     \\ \hline
$T_7=(a,(b,(c,d)))$                & $\frac{(1-\gamma)^2}{3}e^{-y}$     \\ \hline
$T_8=(((b,d),c),a)$                & $\frac{(1-\gamma)^2}{3}e^{-y}$     \\ \hline
$T_9=((a,d),(b,c))$                & 0    \\ \hline
$T_{10}=(((a,b),d),c)$             & 0     \\ \hline
$T_{11}=(b,(a,(c,d)))$             & 0     \\ \hline
$T_{12}=(((a,d),b),c)$             & 0     \\ \hline
$T_{13}=(((b,d),a),c)$             & 0     \\ \hline
$T_{14}=(((a,c),d),b)$             & 0     \\ \hline
$T_{15}=(((a,d),c),b)$             & 0     \\ \hline

\end{tabular}
\end{table}

 \begin{figure}[!ht]
  \centering
    \subfigure[$\gamma=0.5$]{
  	  \label{fig:4taxa_chart:a}
      \includegraphics[width=0.45\textwidth]{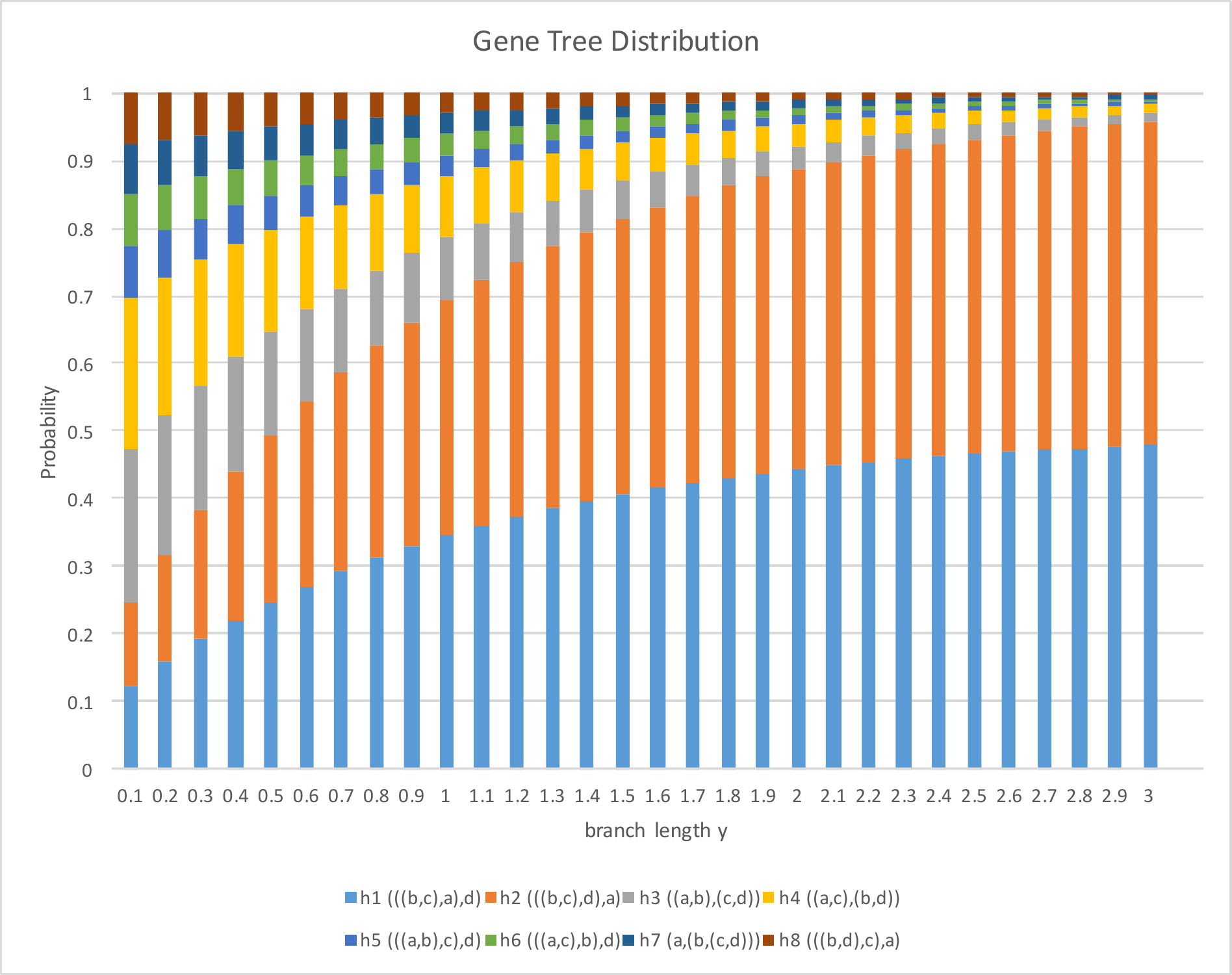}}
    \subfigure[$\gamma=0.05$]{
  	  \label{fig:4taxa_chart:b}
      \includegraphics[width=0.45\textwidth]{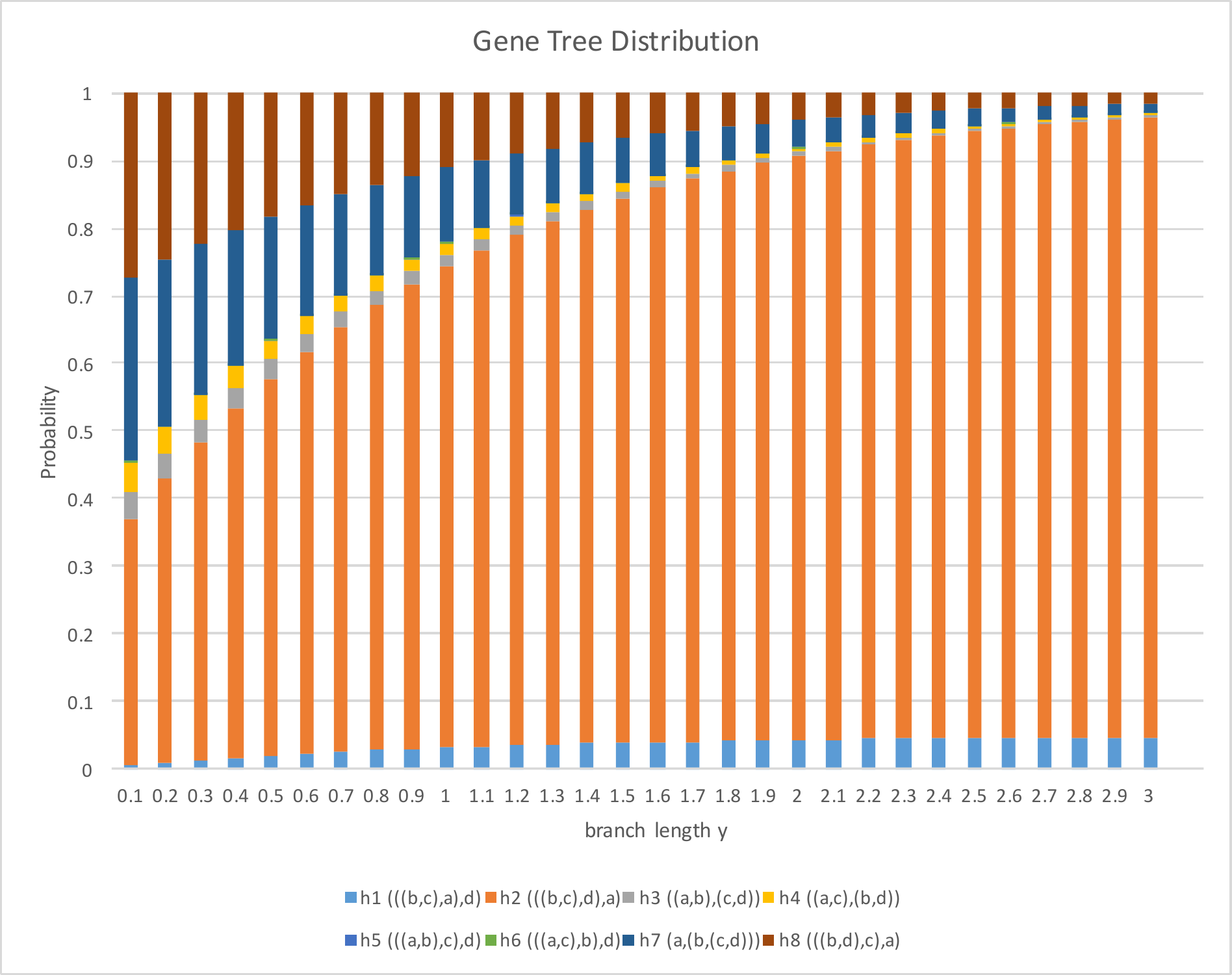}}
  \caption{Gene tree distribution for the phylogenetic network in Figure \ref{fig:dc}(b) ($w=0$) as $x\to\infty$. The x-axis corresponds to branch length $y$ and the y-axis is the probability of each gene tree topology (see Table \ref{tab:table1x} in the Appendix).}
  \label{fig:4taxa_chart}
\end{figure}

\end{document}